\documentclass[twoside,11pt]{article}

\usepackage{jmlr2e}

\usepackage{hyperref}
\usepackage{subfigure}
\usepackage{graphicx}
\usepackage{color}

\usepackage{amsmath}
\usepackage{amssymb} 

\usepackage{bbm}
\usepackage{balance}
\allowdisplaybreaks

\DeclareMathOperator{\pfaffian}{pf}

\ShortHeadings{Tractable Minor-free Generalization of Planar Zero-field Ising Models}{Likhosherstov, Maximov and Chertkov}
\firstpageno{1}

\begin{document}

\title{Tractable Minor-free Generalization of Planar Zero-field Ising Models}

\author{\name Valerii Likhosherstov \email vl304@cam.ac.uk \\
       \addr Department of Engineering \\ University of Cambridge \\ Cambridge, UK
       \AND
       \name Yury Maximov \email yury@lanl.gov \\
       \addr Theoretical Division and Center for Nonlinear Studies \\ Los Alamos National Laboratory \\ Los Alamos, NM, USA \\
       \name Michael Chertkov \email
       chertkov@math.arizona.edu \\
       \addr Graduate Program in Applied Mathematics \\ University of Arizona \\ Tucson, AZ, USA}

\editor{}

\maketitle

\begin{abstract}
We present a new family of zero-field Ising models over N binary variables/spins obtained by consecutive ``gluing'' of planar and $O(1)$-sized components and subsets of at most three vertices into a tree. The polynomial time algorithm of the dynamic programming type for solving exact inference (computing partition function) and exact sampling (generating i.i.d. samples) consists in a sequential application of an efficient (for planar) or brute-force (for $O(1)$-sized) inference and sampling to the components as a black box. To illustrate utility of the new family of tractable graphical models, we first build a polynomial algorithm for inference and sampling of zero-field Ising models over $K_{33}$-minor-free topologies and over $K_5$-minor-free topologies---both are extensions of the planar zero-field Ising models---which are neither genus- no treewidth-bounded. Second, we demonstrate empirically an improvement in the approximation quality of the NP-hard problem of inference over the square-grid Ising model in a node-dependent non-zero ``magnetic" field. 
\end{abstract}


\section{Introduction}

Let $G = (V (G), E(G))$ be an undirected graph with a set of vertices $V(G)$ and a set of normal edges $E(G) \subseteq \binom{V (G)}{2}$ (no loops or multiple edges). We discuss \textit{Ising models} which associate the following probability to each random $N\triangleq |V(G)|$-dimensional binary variable/spin configuration $X \in  \{\pm 1 \}^{N}$: 
\begin{gather}
    \mathbb{P}(X)  \triangleq \frac{\mathcal{W}(X)}{Z}, \label{eq:zfim}
\end{gather}
where 
\begin{gather}
    \mathcal{W}(X) \triangleq  \exp\biggl( \sum_{v \in V(G)} \mu_v x_v + \sum_{ e = \{ v, w \} \in E(G)} J_e x_v x_w \biggr) 
   \quad \text{and}\quad Z \triangleq \sum_{X \in \{ \pm 1 \}^N} \mathcal{W}(X). \label{eq:Z} 
\end{gather}
Here, $\mu = (\mu_v, v \in V(G))$ is a vector of \textit{(magnetic) fields}, $J = (J_e, e \in E(G))$ is a vector of the \textit{(pairwise) spin interactions}, and the normalization constant $Z$, which is defined as a sum over $2^N$ spin configurations, is referred to as the \textit{partition function}. Given the model specification $\mathcal{I} = \langle G, \mu, J \rangle$, we address the tasks of finding the exact value of $Z$ (inference) and drawing exact samples with the probability (\ref{eq:zfim}).

\textbf{Related work.} 
It has been known since the seminal contributions of \citet{fisher} and \citet{kasteleyn} that computation of the partition function in the zero-field ($\mu = 0$) Ising model over a planar graph and sampling from the respective probability distribution are both tractable, that is, these are tasks of complexity polynomial in $N$.
As shown by \citet{barahona}, even when $G$ is planar or when $\mu = 0$ (\textit{zero field}),
the positive results are hard to generalize---both addition of the non-zero (magnetic) field  and the extension beyond planar graphs make
the computation of the partition function NP-hard. These results are also consistent with the statement from \citet{jerrum-sinclair} that computation of the partition function of the zero-field Ising model is a \#P-complete problem, even in the ferromagnetic case when all components of $J$ are positive. Therefore, describing $\langle G, \mu, J \rangle$ families for which computations of the partition function and sampling are tractable remains an open question.

The simplest tractable (i.e., inference and sampling are polynomial in $N$) example is one when $G$ is a tree, and the corresponding inference algorithm, known as \textit{dynamic programming} and/or \textit{belief propagation}, has a long history in physics \citep{bethe,peierls}, optimal control \citep{bellman}, information theory \citep{gallager}, and artificial intelligence \citep{pearl}. Extension to the case when $G$ is a tree of $(t + 1)$-sized cliques ``glued'' together, or more formally when $G$ is of a \textit{treewidth}~$t$, is known as the \textit{junction tree algorithm} \citep{jensen}, which has complexity of counting and sampling that grow exponentially with $t$.

Another insight 
originates from the foundational statistical physics literature of the last century related to the zero-field version of (\ref{eq:zfim}), i.e. when $\mu=0$, over planar $G$. \citet{onsager} found a closed-form solution of (\ref{eq:zfim}) in the case of a homogeneous Ising model over an infinite two-dimensional square grid. \citet{kac-ward} reduced the inference of (\ref{eq:zfim}) over a finite square lattice to computing a determinant. \citet{kasteleyn} generalized this result to an arbitrary (finite) planar graph. Kasteleyn's approach consists of expanding each vertex of $G$ into a gadget and reducing the Ising model inference to the problem of counting perfect matchings over the expanded graph. Kasteleyn's construction was simplified by \citet{fisher}. The tightest running time estimate for Kasteleyn's method gives $O(N^\frac32)$. Kasteleyn conjectured, which was later proven by \citet{gallucio}, that the approach extends to the case of the zero-field Ising model over graphs embedded in a surface of \textit{genus} $g$ with a multiplicative $O(4^g)$ penalty. 

A parallel way of reducing the planar zero-field Ising model to a perfect matching counting problem consists of constructing the so-called expanded dual graph \citep{bieche,barahona,schraudolph-kamenetsky}. 
This approach is advantageous because using the expanded dual graph allows a one-to-one correspondence between spin configurations and perfect matchings. An extra advantage of this approach is that the reduction allows us to develop an exact efficient sampling. Based on linear algebra and planar separator theory \citep{lipton-tarjan}, \citet{wilson} introduced an algorithm that allows to sample perfect matchings over planar graphs in $O(N^\frac32)$ time. The algorithms were implemented by \citet{thomas-middleton1,thomas-middleton2} for the Ising model sampling, however, the implementation was limited to only the special case of a square lattice. \citet{thomas-middleton1} also suggested a simple extension of the Wilson's algorithm to the case of bounded genus graphs, again with the $4^g$ factor in complexity. Notice that imposing the zero field condition is critical, as otherwise, the Ising model over a planar graph is  NP-hard \citep{barahona}. On the other hand, even in the case of zero magnetic field the Ising models over general graphs are difficult \citep{barahona}.

Wagner's theorem \citep[chap.~4.4]{diestel} states that $G$ is planar if and only if it does not have $K_{33}$ and $K_5$ as minors (Figure \ref{fig:k5}(b)). Both families of $K_{33}$-free and $K_5$-free graphs generalize and extend the family of planar graphs, since $K_{33}$ ($K_{5}$) is nonplanar but $K_5$-free ($K_{33}$-free). Both families are genus-unbounded, since a disconnected set of $g$ $K_{33}$ ($K_5$) graphs has a genus of $g$ \citep{battle} and is $K_5$-free ($K_{33}$-free). Moreover, both families are treewidth-unbounded, since planar square grid of size $t \times t$ has a treewidth of $t$ \citep{bodlaender}. Therefore, the question of interest becomes generalizing tractable inference and sampling in the zero-field Ising model over a $K_{33}$-free or $K_5$-free graph.

To extend tractability of the special cases as an approximation to a more general class of inference problems it is natural to consider a family of tractable spanning subgraphs and then exploit the fact that the log-partition function $\log Z (\mu, J)$ is convex and hence can be upper-bounded by a linear combination of tractable partition functions. Tree-reweighted (TRW) approximation \citep{wainwright} was the first example in the literature where such upper-bounding was constructed with the trees used as a basic element. The upper-bound TRW approach \citep{wainwright} was extended by \citet{globerson}, where utilizing a planar spanning subgraph (and not a tree) as the basic (tractable) element was suggested.

\textbf{Contribution.}
In this manuscript, we, first of all, compile results that were scattered over the literature on (at least) $O(N^\frac32)$-efficient exact sampling and exact inference in the zero-field Ising model over planar graphs. To the best of our knowledge,  we are the first to present a complete and mathematically accurate description of the tight asymptotic bounds.

Then, we describe a new family of zero-field Ising models on graphs that are more general than planar. Given a tree decomposition of such graphs into planar and ``small'' ($O(1)$-sized) components ``glued'' together along sets of at most three vertices, inference and sampling over the new family of models is of polynomial time. We further show that all the $K_{33}$-free or $K_5$-free graphs are included in this family and, moreover, their aforementioned tree decomposition can be constructed with $O(N)$ efforts. This allows us to prove an $O(N^\frac32)$ upper bound on run time complexity for exact inference and exact sampling of the $K_5$-free or $K_{33}$-free zero-field Ising models. 



Finally, we show how the newly introduced tractable family of the zero-field Ising models allows extension of the approach of \citet{globerson} resulting in an upper-bound for log-partition function over general Ising models, non-planar and including non-zero magnetic field. Instead of using planar spanning subgraphs as in the work of \citet{globerson}, we use more general (non-planar) basic tractable elements. Using the methodology of \citet{globerson}, we illustrate the approach through experiments with a nonzero-field Ising model on a square grid for which exact inference is known to be NP-hard \citep{barahona}.

\textbf{Relation to other algorithms.} The result presented in this manuscript is similar to the approach used to count perfect matchings in $K_5$-free graphs \citep{curticapean,straub}. However, we do not use a transition to perfect matching counting as it is typically done in studies of zero-field Ising models over planar graphs~\citep{fisher,kasteleyn,thomas-middleton1}. 
Presumably, a direct transition to perfect matching counting can be done via a construction of an expanded graph in the fashion of \citet{fisher,kasteleyn}. However, this results in a size increase and, what's more important, there is no direct correspondence between spin configurations and perfect matchings, therefore exact sampling is not supported.

\textbf{Structure.} Section \ref{ch:planar} states the problems of exact inference and exact sampling for planar zero-field Ising models. In Section \ref{sec:main} we introduce the concept of $c$-nice decomposition of graphs, and then formulate and prove tractability of the zero-field Ising models over graphs which are $c$-nice decomposible.
Section \ref{sec:minorfree} is devoted to application of the algorithm introduced in the preceding Section to examples of the zero-field Ising model over the $K_{33}$-free (but possibly $K-5$ containing) and $K_5$-free (but possibly $K_{3,3}$ containing) graphs. Section \ref{sec:emp} presents an empirical application of the newly introduced family of tractable models to an upper-bounding log-partition function of a broader family of intractable graphical models (planar nonzero-field Ising models). Section \ref{sec:concl} is reserved for conclusions.

Throughout the text, we use common graph-theoretic notations and definitions \citep{diestel} and also restate the most important concepts briefly.

\section{Planar Topology} \label{ch:planar}

In this Section, we consider the special $\mathcal{I} = \langle G, 0, J \rangle$ case of the zero-field Ising model over a planar graph and introduce transition from $\mathcal{I}$ to the perfect matching model over a different (derived from $G$) planar graph. One-to-one correspondence between a spin configuration over the Ising model and corresponding perfect matching configuration over the derived graph translates the exact inference and exact sampling over  $\mathcal{I}$ to the corresponding exact inference and exact sampling in the derived perfect matching model.

\subsection{Expanded Dual Graph} \label{subsec:edg}

The graph is \textit{planar} when it can be drawn on (embedded into) a plane without edge intersections. We assume that the planar embedding of $G$ is given
(and if not, it can be found in $O(N)$ time according to \citet{boyer}). In this Section we follow in our constructions of \citet{schraudolph-kamenetsky}.

Let us, first, triangulate $G$ by triangulating one after another each face of the original graph and then setting $J_e = 0$ for all the newly added edges $e\in E(G)$.  Complexity of the triangulation is $O(N)$, see \citet{schraudolph-kamenetsky} for an example. (For convenience, we will then use the same notation for the derived, triangulated graph as for the original graph.) 

Second, construct a new graph, $G_F$, where each vertex $f$ of $V (G_F)$ is a face of $G$, and there is an edge $e = \{ f_1, f_2 \}$ in $E (G_F)$ if and only if $f_1$ and $f_2$ share an edge in $G$. By construction, $G_F$ is planar, and it is embedded in the same plane as $G$, so that each new edge $e = \{ f_1, f_2 \} \in E (G_F)$ intersects the respective old edge. 
Call $G_F$ a \textit{dual graph} of $G$. Since $G$ is triangulated, each $f \in V (G_F)$ has degree 3 in $G_F$. 

Third, obtain a planar graph $G^*$ and its embedding from $G_F$ by substituting each $f \in V (G_F)$ by a $K_3$ triangle so that each vertex of the triangle is incident to one edge, going outside the triangle (see Figure \ref{fig:expanded_dual} for illustration). Call $G^*$ the \textit{expanded dual graph} of $G$.

Newly introduced triangles of $G^*$, substituting $G_F$'s vertices, are called \textit{Fisher cities} \citep{fisher}. We refer to edges outside triangles as \textit{intercity edges} and denote their set as $E^*_I$. The set $E (G^*) \setminus E^*_I$ of Fisher city edges is denoted as $E^*_C$. Notice that $e^* \in E^*_I$ intersects exactly one $e \in E (G)$ and vice versa, which defines a bijection between $E^*_I$ and $E (G)$; denote it by $g: E^*_I \to E (G)$. Observe that $| E^*_I | = | E (G) | \leq 3 N - 6$, where $N$ is the size (cardinality) of $G$.

A set $E' \subseteq E (G)$ is called a \textit{perfect matching} (PM) of $G$, if edges of $E'$ are disjoint and their union equals $V$. Let $\text{PM}(G)$ denote the set of all Perfect Matchings (PM) of $G$. Notice that $E^*_I$ is a PM of $G^*$, and thus $| V (G^*) | = 2 | E^*_I | = O(N)$. Since $G^*$ is planar, one also finds that $| E (G^*) | = O(N)$. Constructing $G^*$ requires $O(N)$ steps.

\subsection{Perfect Matching (PM) Model} \label{subsec:pmc}

For every spin configuration $X \in \{ \pm 1 \}^N$, let $I(X)$ be a set $\{ e \in E^*_I \, | \, g(e) = \{ v, w \}, x_v = x_w \}$. Each Fisher city is incident to an odd number of edges in $I(X)$. Thus, $I(X)$ can be uniquely completed to a PM by edges from $E^*_C$. Denote the resulting PM by $M(X) \in \text{PM}(G^*)$ (see Figure \ref{fig:pm_cases} for an illustration). Let $\mathcal{C}_+ = \{ +1 \} \times \{ \pm 1 \}^{N - 1}$. 
\begin{lemma} \label{lemma:bij}
$M$ is a bijection between $\mathcal{C}_+$ and $\text{PM}(G^*)$.
\end{lemma}

\begin{figure}
\centering
\subfigure[]{
  \centering
  \includegraphics[width=0.3\linewidth]{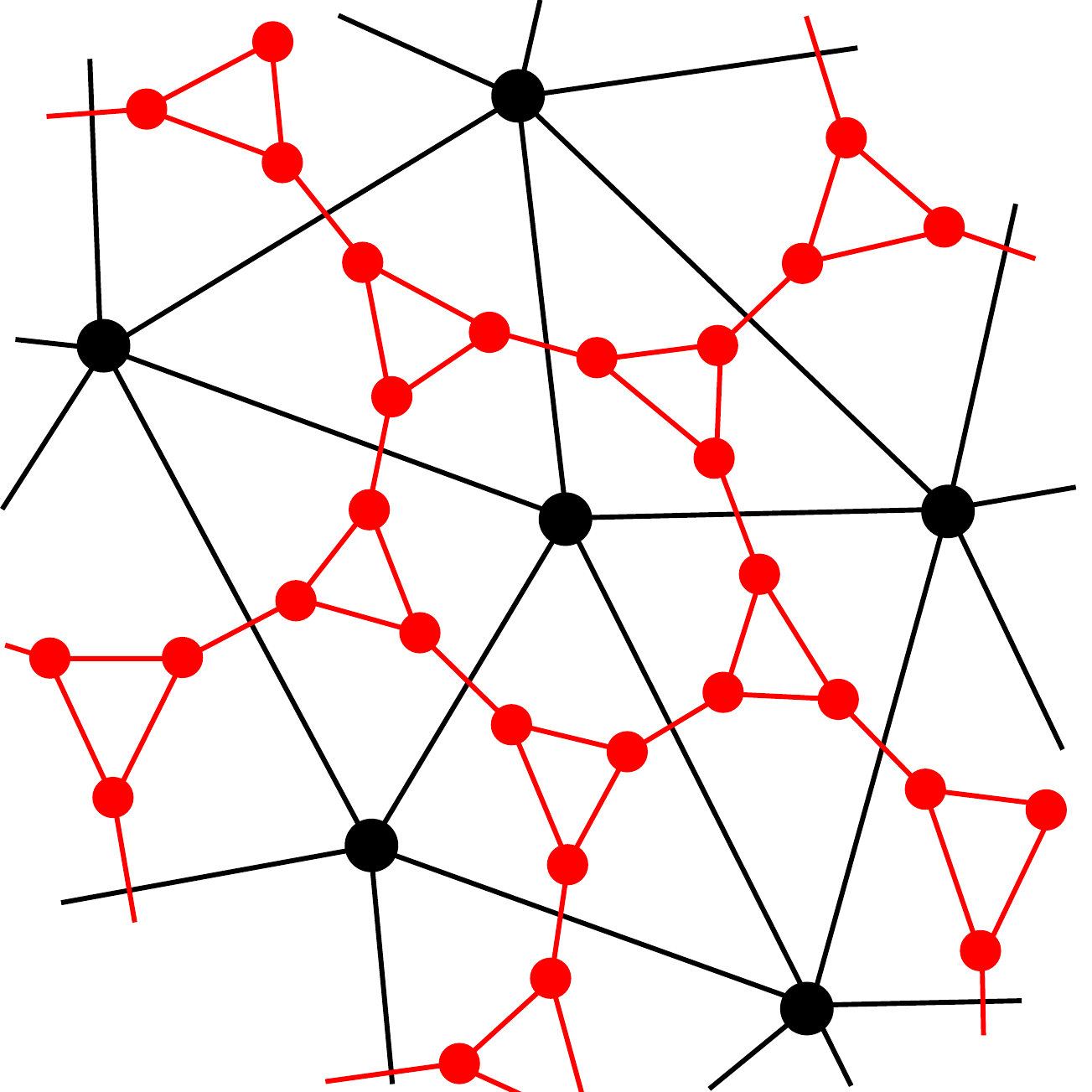}
  \label{fig:expanded_dual}
}%
\subfigure[]{
  \centering
  \includegraphics[width=0.3\linewidth]{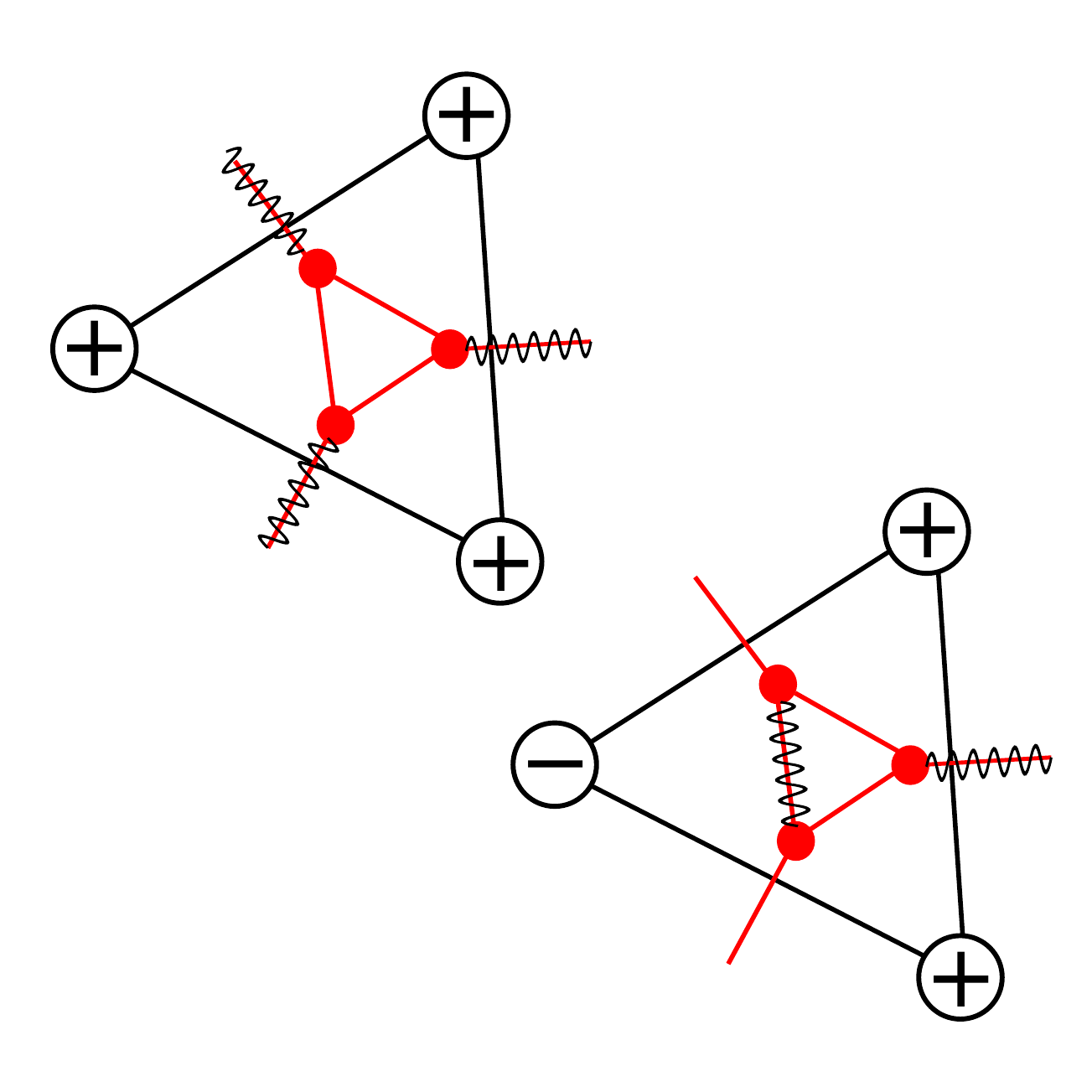}
  \label{fig:pm_cases}
}%
\caption{
(a) A fragment of $G$'s embedding after triangulation (black), expanded dual graph $G^*$ (red). (b) Possible $X$ configurations and corresponding $M(X)$ (wavy lines) on a single face of $G$. Rotation symmetric and reverse sign configurations are omitted.}
\end{figure}


Define weights on $G^*$ according to
\begin{equation}
\forall e^* \in E (G^*): c_{e^*} \triangleq \begin{cases} \exp (2 J_{g(e^*)}), & e^* \in E^*_I \\ 1, & e^* \in E^*_C \end{cases}
\end{equation}

\begin{lemma} \label{lemma:zfitopm}
For $E' \in \text{PM}(G^*)$ holds
\begin{equation}
    \mathbb{P} ( M(X) = E' ) = \frac{1}{Z^*} \prod_{e^* \in E'} c_{e^*}, \label{eq:pmprobs}
\end{equation}
where
\begin{equation}
    Z^* \triangleq \sum_{E' \in \text{PM}(G^*)} \prod_{e^* \in E'} c_{e^*} = \frac12 Z \exp\left( \sum_{e \in E(G)} J_e\right)
    \label{eq:zstar}
\end{equation}
is the partition function of the PM distribution (PM model) defined by (\ref{eq:pmprobs}).
\end{lemma}

See proofs of the Lemma \ref{lemma:bij} and Lemma \ref{lemma:zfitopm} in Appendix \ref{sec:lp}. Second transition of (\ref{eq:zstar}) reduces the problem of computing $Z$ to computing $Z^*$. Furthermore, only two equiprobable spin configurations $X'$ and $-X'$ (one of which is in $\mathcal{C}_+$) correspond to $E'$, and they can be recovered from $E'$ in $O(N)$ steps, thus resulting in the statement that one samples from $\mathcal{I}$ if sampling from (\ref{eq:pmprobs}) is known.


The PM model can be defined for an arbitrary graph $\hat{G}$, $\hat{N} = | V(\hat{G}) |$ with positive weights $c_e, e \in E'$, as a probability distribution over $\hat{M} \in \text{PM}(\hat{G})$: $\mathbb{P}(\hat{M}) \propto \prod_{e \in \hat{M}} c_e$.
Our subsequent derivations are based on the following
\begin{theorem} \label{th:pmmodel}
    Given the PM model defined on planar graph $\hat{G}$ of size $\hat{N}$ with positive edge weights $\{ c_e \}$, one can find its partition function and sample from it in $O(\hat{N}^\frac32)$ time (steps).
\end{theorem}

Algorithms, constructively proving the theorem, are directly inferred from \citet{wilson,thomas-middleton1}, with minor changes/generalizations. We describe the algorithms in Appendix \ref{ch:pl_proof}.

\begin{corollary}
     Exact inference and exact sampling of the PM model over $G^*$ (and, hence, zero-field Ising model $\mathcal{I}$ over the planar graph $G$) take $O(N^\frac32)$ time.
\end{corollary}

\section{\texorpdfstring{$c$}{c}-nice Decomposition of the Topology} \label{sec:main}

We commence by introducing the concept of $c$-nice decomposition of a graph and stating the main result on the tractability of the new family of Ising models in Subsection \ref{sec:dec}. Then we proceed building  a helpful ``conditioning'' machinery in Subsection \ref{sec:cond} and subsequently describing algorithms for the the efficient exact inference (Subsection \ref{sec:inf}) and exact sampling (Subsection \ref{sec:samp}), therefore proving the aforementioned statement constructively.


\subsection{Decomposition tree and the key result (of the manuscript)} \label{sec:dec}

We mainly follow \citet{curticapean,reed} in the definition of the decomposition tree and its properties sufficient for our goals. (Let us also remind that we consider here graphs containing no self-loops or multiple edges.)

Graph $G'$ is a \textit{subgraph} of $G$ whenever $V(G') \subseteq V(G)$ and $E(G') \subseteq E(G)$. For two subgraphs $G'$ and $G''$ of $G$, let $G' \cup G'' = (V(G') \cup V(G''), E(G') \cup E(G''))$ (graph \textit{union}).


Consider a tree decomposition $\mathcal{T} = \langle T, \mathcal{G} \rangle$ of a graph $G$ into a set of subgraphs $\mathcal{G} \triangleq \{ G_t \}$ of $G$, where $t$ are \textit{nodes} of a tree $T$, that is, $t \in V(T)$. One of the nodes of the tree, $r \in V(T)$, is selected as the root. For each node $t \in V(T)$, its \textit{parent} is the first node on the unique path from $t$ to $r$.
$G_{\leq t}$  denotes the graph union of $G_{t'}$ for all the nodes $t'$ in $V(T)$ that are $t$ or its descendants.
$G_{\nleq t}$ denotes the graph union of $G_{t'}$ for all the nodes $t'$ in $V(T)$ that are neither $t$ nor descendants of~$t$.
For two neighboring nodes of the tree, $t, p \in V(T)$ and $\{ t, p \} \in E(T)$, the set of overlapping vertices of $G_t$ and $G_p$, $K \triangleq V(G_t) \cap V(G_p)$, is called an  \textit{attachment set} of~$t$ or~$p$. If $p$ is a parent of~$t$, then $K$ is a \textit{navel} of $t$. We assume that the navel of the root is empty.

$\mathcal{T}$ is a \textit{$c$-nice decomposition} of $G$ if the following requirements are satisfied:
\begin{enumerate}
    \item $\forall t \in V(T)$ with a navel $K$, it holds that $K = V(G_{\leq t}) \cap V(G_{\nleq t})$.

    \item Every attachment set $K$ is of size $0$, $1$, $2$, or $3$.

    \item $\forall t \in V(T)$, either $| V(G_t) | \leq c$ or $G_t$ is planar.

    \item If $t \in V(T)$ is such that $| V(G_t) | > c$, addition of all edges of type $e = \{ v, w \}$, where $v, w$ belong to the same attachment set of $t$ (if $e$ is not yet in $E(G_t)$) does not destroy planarity of $G_t$.
\end{enumerate}


Stating it informally, the $c$-nice decomposition of $G$ is a tree decomposition of $G$ into planar and ``small'' (of size at most $c$) subgraphs $G_t$, ``glued'' via subsets of at most three vertices of $G$. Figure~\ref{fig:k5}(a) shows an example of a $c$-nice decomposition with $c = 8$.
There are various similar ways to define a graph decomposition in literature, and the one presented above is customized (to our purposes) to include only properties significant for our consecutive analysis. 

The remainder of this Section is devoted to a constructive proof of the following key statement of the manuscript.
\begin{theorem} \label{th:main}
    Let $\mathcal{I} = \langle G, 0, J \rangle$ be any zero-field Ising model where there exists a $c$-nice decomposition $\mathcal{T}$ of $G$, where $c$ is an absolute constant. Then, there is an algorithm which, given $\mathcal{I}, \mathcal{T}$ as an input: (1) finds $Z$ and (2) samples a configuration from $\mathcal{I}$ in time $O( \sum\limits_{t \in V(T)} | V(G_t) |^\frac32 )$. 
\end{theorem}



\begin{figure}[!th]
    \centering
    \includegraphics[width=0.9\linewidth]{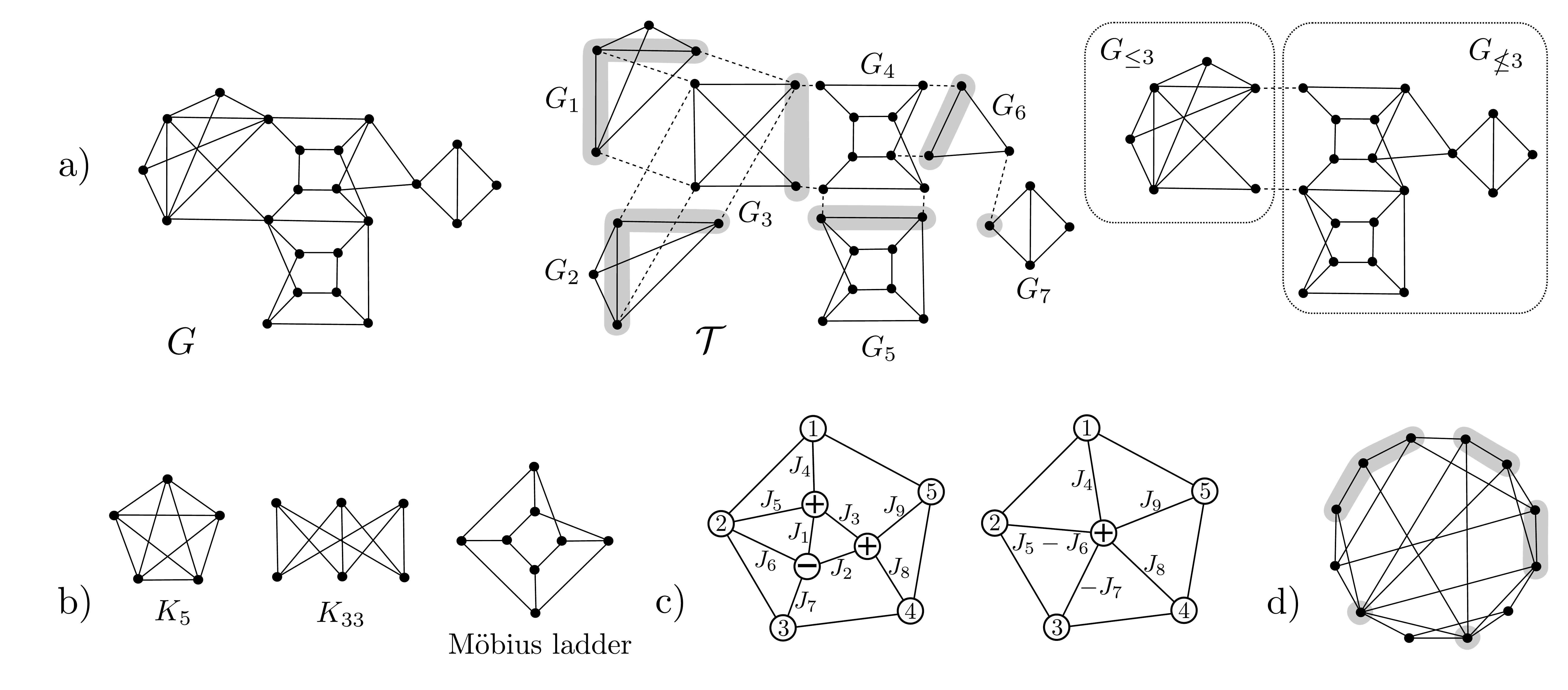}
    \caption{a) An exemplary graph $G$ and its $8$-nice decomposition $\mathcal{T}$, where $t\in \{1, \cdots, 7\}$ labels nodes of the decomposition tree $T$ and node $4$ is chosen as the root ($r = 4$). Identical vertices of $G$ in its subgraphs $G_t$ are shown connected by dashed lines. Navels of size $1$, $2$, and $3$ are highlighted. Component $G_5$ is nonplanar, and $G_4$ becomes nonplanar when all attachment edges are added (according to the fourth item of the definition of the $c$-nice decomposition). $G_{\leq 3}$ and $G_{\nleq 3}$ are shown with dotted lines. Note that the decomposition is non-unique for the graph. For instance, edges that belong  to the attachment set can go to either of the two subgraphs containing this set or even repeat in both. b) Minors $K_5$ and $K_{33}$ are forbidden in the planar graphs. M\"{o}bius ladder and its subgraphs are the only nonplanar graphs allowed in the $8$-nice decomposition of a $K_5$-free graph. c) The left panel is an example of conditioning on three vertices/spins in the center of a graph. The right panel shows a modified graph where the three vertices (from the left panel) are reduced to one vertex, then leading to a modification of the pairwise interactions within the associated zero-field Ising model over the reduced graph. 
    d) Example of a graph that contains $K_5$ as a minor: by contracting the highlighted groups of vertices and deleting the remaining vertices, one arrives at the $K_5$ graph.}
    \label{fig:k5}
\end{figure}

\subsection{Inference and sampling conditioned on 1, 2, or 3 vertices/spins} \label{sec:cond}

Before presenting the algorithm that proves Theorem \ref{th:main} constructively, let us introduce the auxiliary machinery of ``conditioning'', which describes the partition function of a zero-field Ising model over a planar graph conditioned on $1$, $2$, or $3$ spins. 
Consider a zero-field Ising model $\mathcal{I} = \langle G, 0, J \rangle$ defined over a planar graph $G$. We intend to use the algorithm for efficient inference and sampling of $\mathcal{I}$ as a black box in our subsequent derivations.

Let us now introduce the notion of \textit{conditioning}. Consider a spin configuration $X \in \{ \pm 1 \}^N$, a subset $V' = \{ v^{(1)}, \dots, v^{(\omega)} \} \subseteq V(G)$, and define a \textit{condition} $S = \{ x_{v^{(1)}} = s^{(1)}, \dots, x_{v^{(\omega)}} = s^{(\omega)} \}$ \textit{on} $V'$, where $s^{(1)}, \dots, s^{(\omega)} = \pm 1$ are fixed values. Conditional versions of the probability distribution (\ref{eq:zfim}--\ref{eq:Z}) and the \textit{conditional} partition function become 
\begin{eqnarray}
    && \mathbb{P}(X | S) \triangleq \frac{\mathcal{W}(X) \times \mathbbm{1}(X | S)}{Z_{|S}}, \quad \mathbbm{1} (X | S) \triangleq 
    \left\{ \begin{array} {cc}1, & x_{v^{(1)}} = s^{(1)}, \dots, x_{v^{(\omega)}} = s^{(\omega)} \\ 0, & \mbox{otherwise} \end{array} \right.
, 
    \label{eq:zfim_cond} \\ &&
    \text{where}~Z_{|S} \triangleq \sum_{X \in \{\pm 1\}^N} \mathcal{W}(X) \times \mathbbm{1} (X | S). \label{eq:Z_cond}
\end{eqnarray}

Notice that when $\omega = 0$, $S = \{ \}$ and (\ref{eq:zfim_cond}--\ref{eq:Z_cond}) is reduced to (\ref{eq:zfim}--\ref{eq:Z}). The subset of $V(G)$ is \textit{connected} whenever the subgraph, induced by this subset is connected. Inference and sampling of $\mathcal{I}$ can be extended as follows (a formal proof can be found in the Appendix \ref{sec:lp}).
\begin{lemma} \label{lemma:cond}
Given $\mathcal{I} = \langle G, 0, J \rangle$ where $G$ is planar and a condition $S$ on a connected subset $V' \subseteq V(G)$, $| V' | \leq 3$, computing the conditional partition function $Z_{|S}$ and sampling from $\mathbb{P}(X | S)$
are tasks of $O(N^\frac32)$ complexity.
\end{lemma}

Intuitively, the conditioning algorithm proving the Lemma takes the subset of connected vertices and ``collapses'' them into a single vertex. The graph remains planar and the task is reduced to conditioning on one vertex, which is an elementary operation given the algorithm from section \ref{ch:planar}.  (See Figure \ref{fig:k5}(c) for an illustration.)

\subsection{Inference algorithm} \label{sec:inf}


This subsection constructively proves the inference part of Theorem \ref{th:main}.
For each $t \in V(T)$, let $\mathcal{I}_{\leq t} \triangleq \langle G_{\leq t}, 0, \{ J_e \, | \, e \in E(G_{\leq t}) \subseteq E(G) \} \rangle$ denote a zero-field Ising \textit{submodel} \textit{induced} by $G_{\leq t}$. Denote the partition function and subvector of $X$ related to $\mathcal{I}_{\leq t}$ as $Z^{\leq t}$ and $X_{\leq t} \triangleq \{ x_v | v \in V(G_{\leq t}) \}$, respectively.

Further, let $K$ be $t$'s navel and let $S = \{ \forall v \in K: x_v = s^{(v)} \}$ denote some condition on $K$.
Recall that $| K | \leq 3$.
For each $t$, the algorithm computes conditional partition functions $Z^{\leq t}_{| S}$ for all choices of condition spin values $\{ s^{(v)} = \pm 1 \}$.
Each $t$ is processed only when its children have already been processed, so the algorithm starts at the leaf and ends at the root. If $r \in G(T)$ is a root, its navel is empty and $G_{\leq r} = G$, hence $Z = Z^{\leq r}_{| \{  \}}$ is computed after $r$'s processing.

Suppose all children of $t$, $c_1, ..., c_m \in V(T)$ with navels $K_1, ..., K_m \subseteq V(G_t)$ have already been processed, and now $t$ itself is considered.  Denote a spin configuration on $G_t$ as $Y_t \triangleq \{ y_v = \pm 1 \, | \, v \in V(G_t) \}$.
$\mathcal{I}_{\leq c_1}, ..., \mathcal{I}_{\leq c_m}$ are $\mathcal{I}_{\leq t}$'s submodels induced by $G_{\leq c_1}, ..., G_{\leq c_m}$, which can only intersect at their navels in $G_t$. Based on this, one states the following dynamic programming relation:
\begin{equation}
Z^{\leq t}_{| S} = \sum_{Y_t \in \{ \pm 1 \}^{| V(G_t) |} } \mathbbm{1} (Y_t | S) \exp \left( \sum_{e = \{ v, w \} \in E(G_t)} J_e y_v y_w \right) \cdot \prod_{i = 1}^m 
Z^{\leq c_i}_{| S_i [ Y_t ]}.
    \label{eq:dp}
\end{equation}
Here, $S_i [ Y_t ]$ denotes a condition $\{ \forall v \in K_i: x_v = y_v \}$ on $K_i$. The goal is to efficiently perform summation in~\eqref{eq:dp}. Let $I^{(0)}, I^{(1)}, I^{(2)}, I^{(3)}$ be a partition of $\{ 1, ..., m \}$ by navel sizes. Figure \ref{fig:infsamp}(a,b) illustrates inference in $t$.

\begin{enumerate}
    \item \textbf{Navels of size 0, 1.} Notice that if $i \in I^{(0)}$, then $Z^{\leq c_i}_{| \{  \}} = Z^{\leq c_i}$ is constant, which was computed before. The same is true for $i \in I^{(1)}$ and $Z^{\leq c_i}_{| S^{(i)} [Y_t]} = \frac{1}{2} Z^{\leq c_i}$.
%
    \item \textbf{Navels of size 2.} Let $i \in I^{(2)}$ denote $K_i = \{ u^i, q^i \}$ and simplify notation $Z^{\leq c_i}_{y_1, y_2} \triangleq Z^{\leq c_i}_{x_{u^i} = y_1, x_{q^i} = y_2}$ for convenience. Notice that $Z^{\leq c_i}_{|S_i [ Y_t ]}$ is strictly positive, and due to the zero-field nature of $\mathcal{I}_{\leq c_i}$, one finds $Z^{\leq c_i}_{| +1, +1} = Z^{\leq c_i}_{| -1, -1}$ and $Z^{\leq c_i}_{| +1, -1} = Z^{\leq c_i}_{| -1, +1}$. Then, one arrives at $\log Z^{\leq c_i}_{| S_i [ Y_t ]} = A_i + B_i y_{u^i} y_{q^i} $, where $A_i \triangleq \log Z^{\leq c_i}_{| +1, +1} + \log Z^{\leq c_i}_{| +1, -1}$ and $B_i \triangleq \log Z^{\leq c_i}_{| +1, +1} - \log Z^{\leq c_i}_{| +1, -1}$.
    \item \textbf{Navels of size 3.} Let $i \in I^{(3)}$, and as above, denote $K_i = \{ u^i, q^i, h^i \}$ and $Z^{\leq c_i}_{y_1, y_2, y_3} \triangleq Z^{\leq c_i}_{x_{u^i} = y_1, x_{q^i} = y_2, x_{h^i} = y_3}$. 
   \begin{figure}
    \centering
    \includegraphics[width=\linewidth]{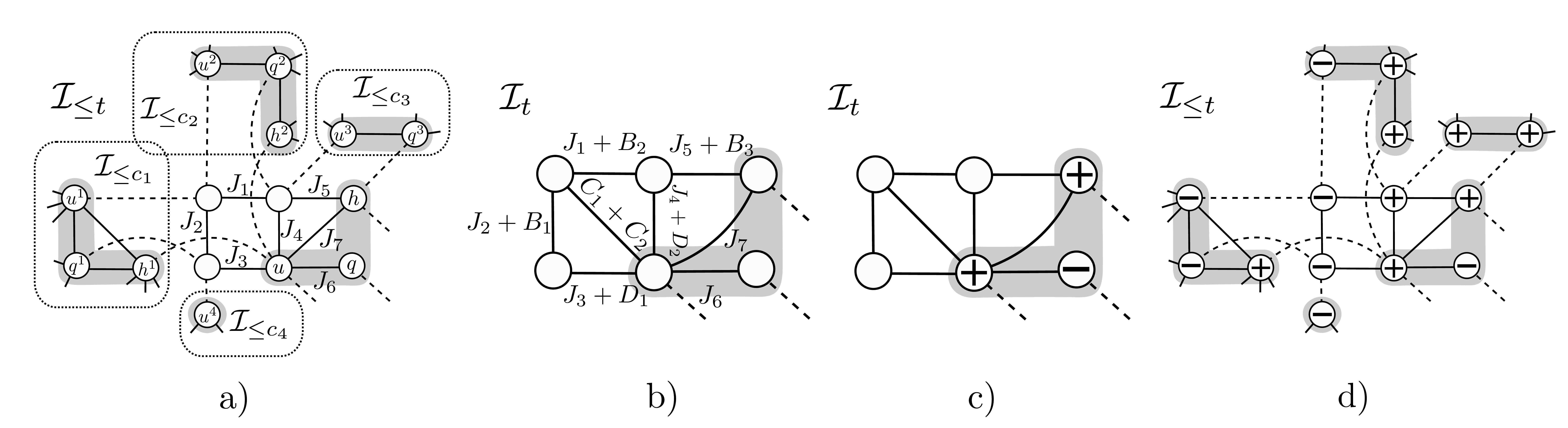}
    \caption{a) Example of inference at node $t$ with children $c_1, c_2, c_3, c_4$. Navels $K_1 = \{ u^1, q^1, h^1 \}, K_2 = \{ u^2, q^2, h^2 \}, K_3 = \{ u^2, q^2 \}, K_4 = \{ u^4 \}$, and $K = \{ u, q, h \}$ are highlighted. Fragments of $\mathcal{I}_{\leq c_i}$ are shown with dotted lines. Here, $I^{(0)} = \varnothing, I^{(1)} = \{ 4 \}, I^{(2)} = \{ 3 \}, \text{and}~ I^{(3)} = \{ 1, 2 \}$, indicating that one child is glued over one node, one child is glued over two nodes, and two children are glued over three nodes. b) ``Aggregated'' Ising model $\mathcal{I}_t$ and its pairwise interactions are shown. Both c) and d) illustrate sampling over $\mathcal{I}_t$. One sample spins in $\mathcal{I}_t$ conditioned on $S^{(t)}$ and then repeats the procedure at the child nodes.}
    \label{fig:infsamp}
\end{figure} 
    
    Due to the zero-field nature of $\mathcal{I}_{\leq c_i}$, it holds that $Z^{\leq c_i}_{| +1, y_2, y_3} = Z^{\leq c_i}_{| -1, y_2, y_3}$. Observe that there are such $A_i, B_i, C_i, D_i$ that $ \log Z^{\leq c_i}_{| y_1, y_2, y_3} = A_i + B_i 
    y_1 y_2+ C_i y_1 y_3 + D_i y_2 y_3$ for all $y_1, y_2, y_3=\pm 1$, which is guaranteed since the following system of equations has a solution:
    \begin{equation}
        \begin{bmatrix} \log Z^{\leq c_i}_{|+1, +1, +1} \\ \log Z^{\leq c_i}_{|+1, +1, -1 } \\ \log Z^{\leq c_i}_{|+1, -1, +1} \\ \log Z^{\leq c_i}_{|+1, -1, -1} \end{bmatrix} = \begin{bmatrix} +1 & +1 & +1 & +1 \\ +1 & +1 & -1 & -1 \\ +1 & -1 & +1 & -1 \\ +1 & -1 & -1 & +1 \end{bmatrix} \times \begin{bmatrix} A_i \\ B_i \\ C_i \\ D_i \end{bmatrix}.
        \label{eq:3vsystem}
    \end{equation}
\end{enumerate}

Considering three cases, one rewrites Eq.~(\ref{eq:dp}) as
\begin{align}
    Z^{\leq t}_{| S} &= M \cdot \sum_{Y_t} \mathbbm{1} (Y_t | S) \exp \biggl( \sum_{e = \{ v, w \} \in E(G_t)} J_e y_v y_w + \sum_{i \in I^{(2)} \cup I^{(3)}} B_i y_{u^i} y_{q^i} \nonumber \\
    &+ \sum_{i \in I^{(3)}} ( C_i y_{u^i} y_{h^i} + D_i y_{q^i} y_{h^i}) \biggr),
    \label{eq:dp4}
\end{align}
where $M \triangleq 2^{- | I^{(1)} |} \cdot \left( \prod_{i \in I^{(0)} \cup I^{(1)}} Z^{\leq c_i} \right) \cdot \exp( \sum_{i \in I^{(2)} \cup I^{(3)}} A_i)$. The sum in Eq.~(\ref{eq:dp4}) is simply a conditional partition function of a zero-field Ising model $\mathcal{I}_t$ defined over a graph $G_t$ with pairwise interactions of $\mathcal{I}$ adjusted by the addition of $B_i, C_i, \text{and}~ D_i$ summands at the appropriate navel edges (if a corresponding edge is not present in $G_t$, it has to be added).
If $| V(G_t) | \leq c$, then (\ref{eq:dp4}) is computed a maximum of four times (depending on navel size) by brute force ($O(1)$ time). Otherwise, if $K$ is a disconnected set in $G_t$, we add zero-interaction edges inside it to make it connected. Possible addition of edges inside $K, K_1, \dots, K_m$ doesn't destroy planarity according to the fourth item in the definition of the $c$-nice decomposition above. Finally, we compute (\ref{eq:dp4}) using Lemma \ref{lemma:cond} in time $O( | V(G_t) |^\frac32 )$.


The inference part of Theorem \ref{th:main} follows directly from the procedure just described.

\subsection{Sampling algorithm} \label{sec:samp}

Next, we address the sampling part of Theorem \ref{th:main}.
We extend the algorithm from section \ref{sec:inf} so that it supports efficient sampling from $\mathcal{I}$.
Assume that the inference pass through $T$ (from leaves to root) has been done so that $\mathcal{I}_t$ for all $t \in V(T)$ are computed. 
Denote $X_t \triangleq \{ x_v \, | \, v \in V(G_t) \}$.
The sampling algorithm runs backwards, first drawing spin values $X_r$ at the root $r$ of $T$ from the marginal distribution $\mathbb{P}(X_r)$, and then processing each node $t$ of $T$ after its parent $p$ is processed. Processing consists of drawing spins $X_t$ from $\mathbb{P}(X_t \, | \, X_p) = \mathbb{P}(X_t \, | \, X^{(t)} \triangleq \{ x_v \, | \, v \in K \})$,
where $K$ is a navel of $t$. This marginal-conditional scheme generates the correct sample $X$ of spins over $G$.

Let 
$\mathbb{P}_{\leq t} (X_{\leq t})$ define a spin distribution of $\mathcal{I}_{\leq t}$. Because the Ising model is an example of Markov Random Field, it holds that $\mathbb{P}_{\leq t} (X_{\leq t} \, | \, X^{(t)} ) = \mathbb{P} (X_{\leq t} \, | \, X^{(t)} )$. We further derive 
\begin{align}
    &\mathbb{P}(X_t \, | \, X^{(t)}) = \mathbb{P}_{\leq t}(X_t \, | \, X^{(t)}) = \frac{1}{Z^{\leq t}} \sum_{X_{\leq t}\setminus X_t} \exp \biggl( \sum_{e = \{ v, w \} \in E(G_{\leq t})} J_e x_v x_w \biggr) \nonumber \\
    &= \frac{1}{Z^{\leq t}} \cdot \exp \biggl( \sum_{e = \{ v, w \} \in E(G_t)} J_e x_v x_w \biggr) \cdot \prod_{i = 1}^m Z^{\leq c_i}_{|S_i [X_t]} \nonumber \\
    &\propto \exp \biggl( \sum_{e = \{ v, w \} \in E(G_t)} J_e x_v x_w + \sum_{i \in I^{(2)} \cup I^{(3)}} B_i x_{u^i} x_{q^i} + \sum_{i \in I^{(3)}} ( C_i x_{u^i} x_{h^i} + D_i x_{q^i} x_{h^i}) \biggr). 
\end{align}

In other words, sampling from $\mathbb{P}(X_t \, | \, X^{(t)})$ is reduced to sampling from $\mathcal{I}_t$ conditional on spins $X^{(t)}$ in the navel $K$. It is done via brute force if $| V(G_t) | \leq c$; otherwise, Lemma \ref{lemma:cond} allows one to draw $X_t$ in $O(| V(G_t) |^\frac32)$, since $| K | \leq 3$. Sampling efforts cost as much as inference, which concludes the proof of Theorem \ref{th:main}. Figure \ref{fig:infsamp}(c,d) illustrates sampling in $t$.

\section{Minor-free Extension of Planar Zero-field Ising Models} \label{sec:minorfree}

\textit{Contraction} is an operation of removing two adjacent vertices $v$ and $u$ (and all edges incident to them) from the graph and adding a new vertex $w$ adjacent to all neighbors of $v$ and $u$.
For two graphs $G$ and $H$, $H$ is $G$'s \textit{minor}, if it is isomorphic to a graph obtained from $G$'s subgraph by a series of contractions (Figure \ref{fig:k5}(d)). $G$ is \textit{$H$-free}, if $H$ is not $G$'s minor.

According to Wagner's theorem \citep[chap.~4.4]{diestel}, a set of planar graphs coincides with an intersection of $K_{33}$-free graphs and $K_5$-free graphs. Some nonplanar graphs are $K_{33}$-free ($K_5$-free), for example, $K_5$ ($K_{33}$). $K_{33}$-free ($K_5$-free) graphs are neither genus-bounded (a disconnected set of $g$ $K_5$ ($K_{33}$) graphs is  $K_{33}$-free ($K_5$-free) and has a genus of $g$ \citep{battle}). $K_{33}$-free ($K_5$-free) graphs are treewidth-unbounded as well (planar square grid of size $t \times t$ is $K_{33}$-free and $K_5$-free and has a treewidth of $t$ \citep{bodlaender}).
In the remainder of the section we show that a $c$-nice decomposition of $K_{33}$-free graphs and $K_5$-free graphs can be computed in polynomial time and, hence, inference and sampling of zero-field Ising models on these graph families can be performed efficiently.

\subsection{Zero-field Ising Models over \texorpdfstring{$K_{33}$}{K33}-free Graphs}

Even though $K_{33}$-free graphs are Pfaffian-orientable (with the Pfaffian orientation computable in polynomial time, see \citet{vazirani}), the expanded dual graph---introduced to map the zero-field Ising model to the respective PM problem---is not necessarily $K_{33}$-free. Therefore, the latter is generally not Pfaffian-orientable. Hence, the reduction to a well-studied perfect matching counting problem is not straightforward.

\begin{theorem} \label{th:k33dec}
Let $G$ be $K_{33}$-free graph of size $N$ with no loops or multiple edges. Then the $5$-nice decomposition $\mathcal{T}$ of $G$ exists and can be computed in time $O(N)$.
\end{theorem}
\begin{proof}(Sketch)
    An equivalent decomposition is constructed by \citet{hopcroft2,gutwenger,vo} in time $O(N)$. We put a formal proof into Appendix \ref{ch:k33dec}.
\end{proof}

\begin{remark}
    The $O(N)$ construction time of $\mathcal{T}$ guarantees that $\sum_{t \in V(T)} | V(G_t) | = O(N)$. All nonplanar components in $\mathcal{T}$ are isomorphic to $K_5$ or its subgraph.
\end{remark}



Therefore, if $G$ is $K_{33}$-free, it satisfies all the conditions needed for efficient inference and sampling, described in section \ref{sec:main}.

\begin{theorem} \label{th:k33comp}
For any $\mathcal{I} = \langle G, 0, J \rangle$ where $G$ is $K_{33}$-free, inference or sampling of $\mathcal{I}$ takes $O(N^\frac32)$ steps.
\end{theorem}
\begin{proof}
    Finding $5$-nice $\mathcal{T}$ for $G$ is the $O(N)$ operation. Provided with $\mathcal{T}$, inference and sampling take at most
    \begin{equation}
        O\left(\sum_{t \in V(T)} | V(G_t) |^\frac32\right) = O\left(\left(\sum_{t \in V(T)} | V(G_t) |\right)^\frac32\right) = O(N^\frac32)
    \end{equation}
    where we apply convexity of $f(z) = z^\frac32$ and the Remark after Theorem \ref{th:k33dec}.
\end{proof}

\subsection{\texorpdfstring{$K_{33}$}{K33}-free Zero-field Ising Models: Implementation and Tests} \label{sec:imp}

In addition to theoretical justification, which is fully presented in this manuscript, we perform emprical simulations to validate correctness of inference and sampling algorithm for $K_{33}$-free zero-field Ising models.

To test the correctness of inference, we generate random $K_{33}$-free models of a given size and then compare the value of PF computed in a brute force way (tractable for sufficiently small graphs) and by our algorithm. See the graph generation algorithm in Appendix \ref{sec:grgen}. We simulate samples of sizes from $\{ 10, ..., 15 \}$ ($1000$ samples per size) and verify that respective expressions coincide.

When testing sampling implementation, we take for granted that the produced samples do not correlate given that the sampling procedure accepts the Ising model as input and uses independent random number generator inside. The construction does not have any memory, therefore, it generates statistically independent samples. To test that the empirical distribution is approaching a theoretical one (in the limit of the infinite number of samples), we draw different numbers $m$ of samples from a model of size $N$. Then we find Kullback-Leibler divergence between the probability distribution of the model (here we use our inference algorithm to compute the normalization, $Z$) and the empirical probability, obtained from samples. Fig.~\ref{fig:kl} shows that KL-divergence converges to zero as the sample size increases. Zero KL-divergence corresponds to equal distributions.

Finally, we simulate inference and sampling for random models of different size $N$ and observe that the computational time (efforts) scales as $O(N^\frac32)$ (Figure~\ref{fig:time}).\footnote{Implementation of the algorithms is available at \href{https://github.com/ValeryTyumen/planar_ising}{https://github.com/ValeryTyumen/planar\_ising}.}

\begin{figure}
\centering
\subfigure[]{
  \centering
  \includegraphics[width=0.45\linewidth]{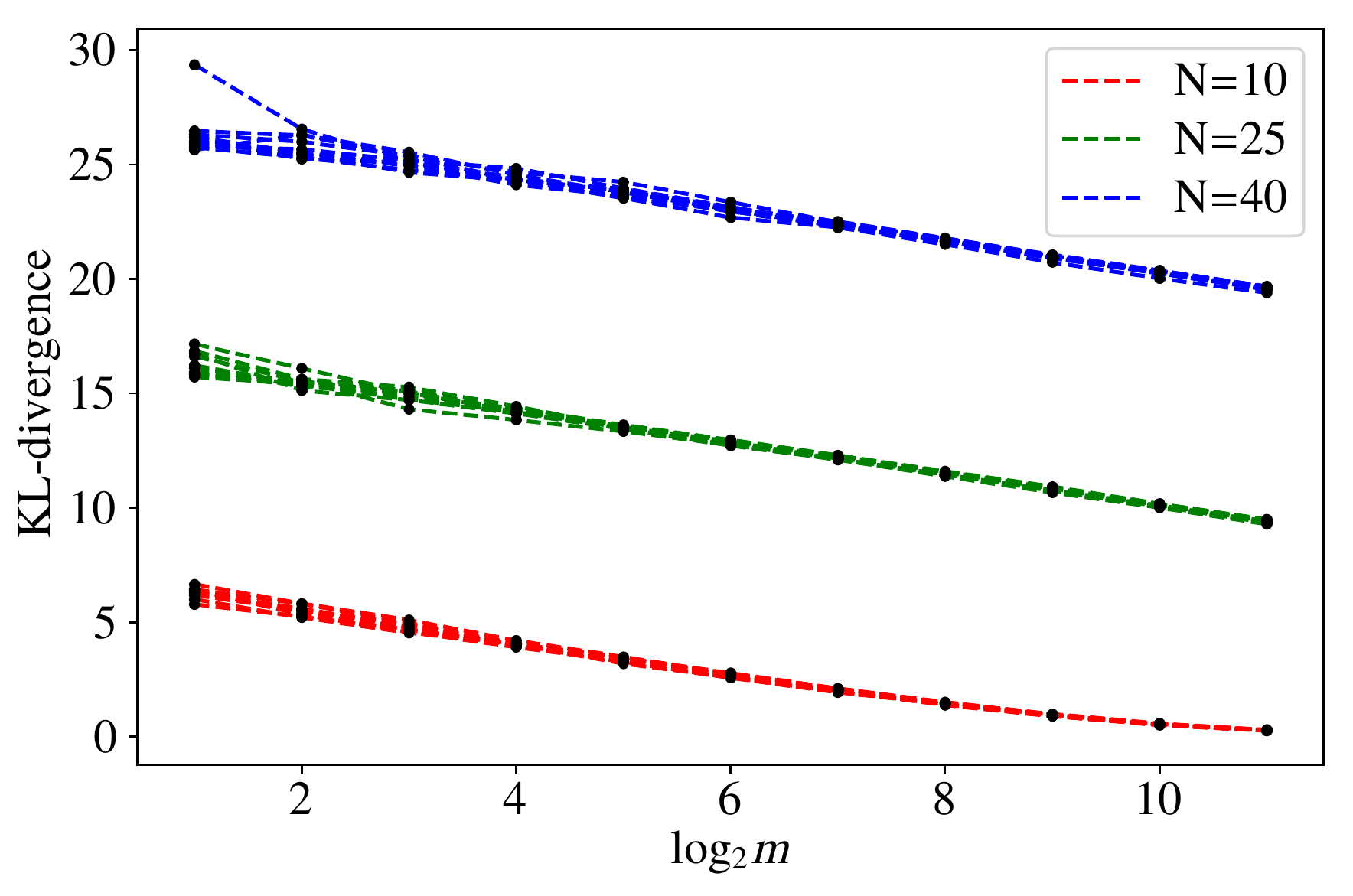}
  \label{fig:kl}
}%
\subfigure[]{
  \centering
  \includegraphics[width=0.45\linewidth]{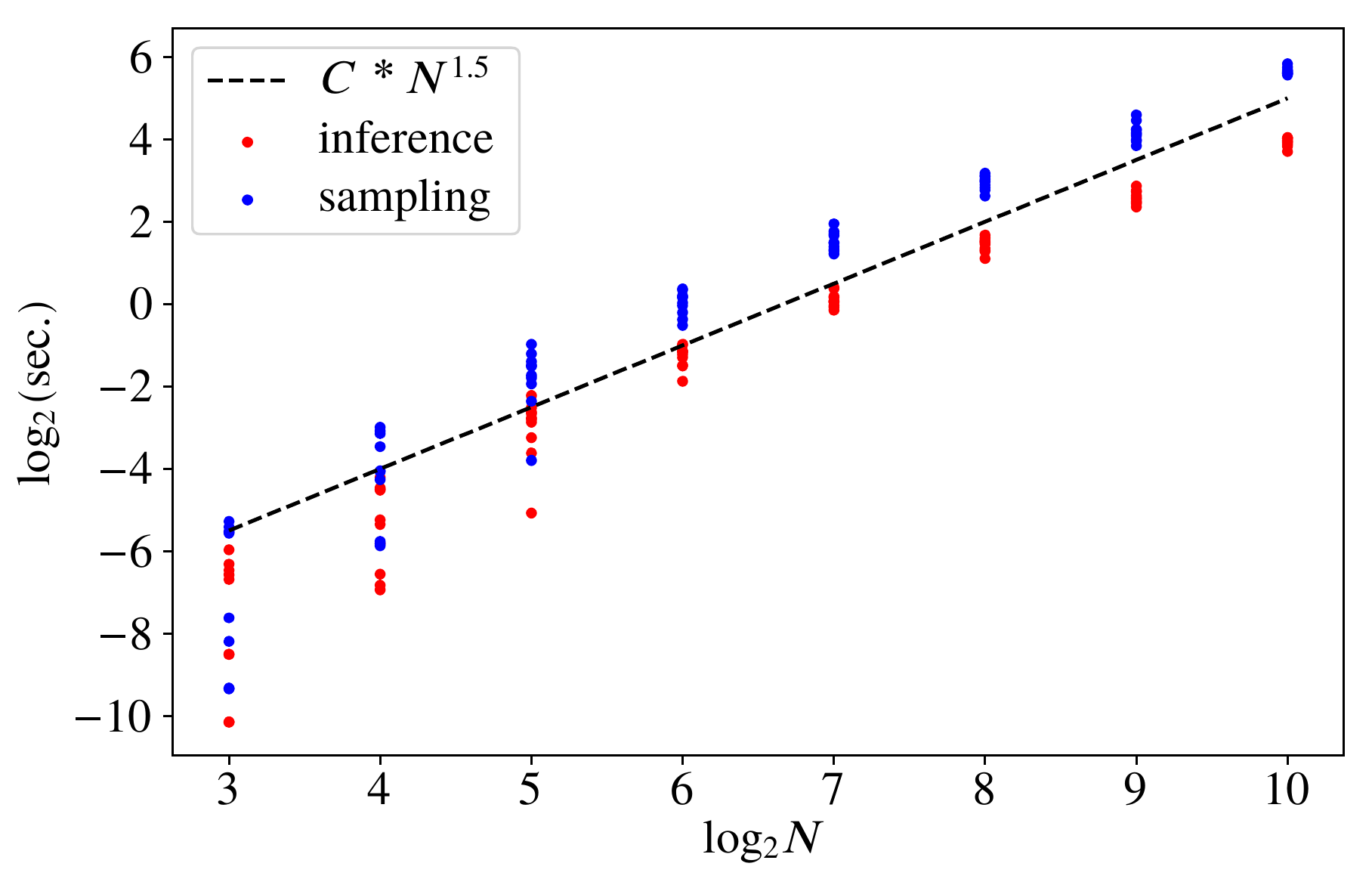}
  \label{fig:time}
}%
\caption{
(a) KL-divergence of the model probability distribution compared with the empirical probability distribution. $N, m$ are the model's size and the number of samples, respectively. (b) Execution time of inference (red dots) and sampling (blue dots) depending on $N$, shown on a logarithmic scale. Black line corresponds to $O(N^\frac32)$.}
\end{figure}

\subsection{Zero-field Ising Models over \texorpdfstring{$K_5$}{K5}-free Graphs} \label{sec:k5}

It can be shown that result similar to the one described above for the $K_{33}$-free graphs also holds for the $K_5$-free graphs as well.
\begin{theorem} \label{th:k5dec}
Let $G$ be a $K_5$-free graph of size $N$ with no loops or multiple edges. Then, the $8$-nice decomposition $\mathcal{T}$ of $G$ exists and can be computed in time $O(N)$.
\end{theorem}
\begin{proof}(Sketch)
    An equivalent decomposition is constructed by \citet{reed} in time $O(N)$. See Appendix \ref{sec:k5proof} for formal proof.
\end{proof}
\begin{remark}
    The $O(N)$ construction time of $\mathcal{T}$ guarantees that 
    $\sum_{t \in V(T)} | V(G_t) | = O(N)$. All nonplanar components in $\mathcal{T}$ are isomorphic to the M\"{o}bius ladder (Figure \ref{fig:k5}(b)) or its subgraph.
\end{remark}

The graph in Figure \ref{fig:k5}(a) is actually $K_5$-free. Theorems \ref{th:main} and \ref{th:k5dec} allow us to conclude:
\begin{theorem} \label{th:inf}
    Given $\mathcal{I} = \langle G, 0, J \rangle$ with $K_5$-free $G$ of size $N$, finding $Z$ and sampling from $\mathcal{I}$ take $O(N^\frac32)$ total time.
\end{theorem}
\begin{proof}
    Analogous to the proof of Theorem \ref{th:k33comp}.
\end{proof}


\section{Approximate Inference of Square-grid Ising Model} \label{sec:emp}

In this section, we consider $\mathcal{I} = \langle G, \mu, J \rangle$ such that $G$ is a square-grid graph of size $H \times H$. Finding $Z (G, \mu, J)$ for arbitrary $\mu$, $J$ is an NP-hard problem \citep{barahona} in such a setting. Construct $G'$ by adding an \textit{apex} vertex connected to all $G$'s vertices by edge (Figure \ref{fig:grid}(a)). Now it can easily be seen that $Z (G, \mu, J) = \frac{1}{2} Z(G', 0, J' = (J_\mu \cup J))$, where $J_\mu = \mu$ are interactions assigned for apex edges.

Let $\{ G^{(r)} \}$ be a family of spanning graphs ($V(G^{(r)}) = V(G')$, $E(G^{(r)}) \subseteq E(G')$) and $J^{(r)}$ be interaction values on $G^{(r)}$. Also, denote $\hat{J}^{(r)} = J^{(r)} \cup \{ 0, e \in E(G') \setminus E(G^{(r)}) \}$. Assuming that $\log Z (G^{(r)}, 0, J^{(r)})$ are tractable, the convexity of $\log Z (G', 0, J')$ allows one to write the following upper bound:
\begin{equation} \label{eq:ub}
    \log Z (G', 0, J') \leq \min_{\substack{\rho(r) \geq 0, \sum_r \rho(r) = 1 \\ \{ J^{(r)} \}, \sum_r \rho(r) \hat{J}^{(r)} = J'}} \sum_r \rho(r) \log Z(G^{(r)}, 0, J^{(r)}).
\end{equation}

After graph set $\{ G^{(r)} \}$ has been fixed, one can numerically optimize the right-hand side of (\ref{eq:ub}), as shown in \citet{globerson} for planar $G^{(r)}$.  The extension of the basic planar case is straightforward and is detailed in the Appendix \ref{sec:mcas}. The Appendix also contains description of marginal probabilities approximation suggested in \citet{globerson,wainwright}.

The choice for a planar spanning graph (PSG) family $\{ G^{(r)} \}$ of \citet{globerson} is illustrated in Figure \ref{fig:grid}(b). A tractable decomposition-based extension of the planar case presented in this manuscript suggests a more advanced construction---decomposition-based spanning graphs (DSG) (Figure \ref{fig:grid}(c)). We compare performance of both PSG and DSG approaches as well as the performance of tree-reweighted approximation (TRW) \citep{wainwright} in the following setting of \textit{Varying Interaction}: $\mu \sim \mathcal{U}(-0.5, 0.5)$, $J \sim \mathcal{U}(-\alpha, \alpha)$, where $\alpha \in \{ 1, 1.2, 1.4, \dots, 3 \}$. We opt optimize for grid size $H = 15$ ($225$ vertices, $420$ edges) and compare upper bounds and marginal probability approximations (superscript \textit{alg}) with exact values obtained using a junction tree algorithm \citep{jensen} (superscript \textit{true}). We compute three types of error:
\begin{enumerate}
    \item normalized log-partition error $\frac{1}{H^2} (\log Z^{alg} - \log Z^{true})$,
    \item error in pairwise marginals $\frac{1}{| E(G) |} \sum_{e = \{ v, w \} \in E(G)} | \mathbb{P}^{alg} (x_v x_w = 1) - \mathbb{P}^{true} (x_v x_w = 1) |$, and
    \item error in singleton central marginal $| \mathbb{P}^{alg} (x_v = 1) - \mathbb{P}^{true} (x_v = 1)|$ where $v$ is a vertex of $G$ with coordinates $(8, 8)$.
\end{enumerate}

We average results over $100$ trials (see Fig.~\ref{fig:plot}).\footnote{Hardware used: 24-core Intel\textregistered \, Xeon\textregistered \, Gold 6136 CPU @ 3.00 GHz}\footnote{Implementation of the algorithms is available at \url{https://github.com/ValeryTyumen/planar\_ising}} We use the same quasi-Newton algorithm \citep{bertsekas} and parameters when optimizing (\ref{eq:ub}) for PSG and DSG, but for most settings, DSG outperforms PSG and TRW. Cases with smaller TRW error can be explained by the fact that TRW implicitly optimizes~\eqref{eq:ub}  over the family of \textit{all} spanning trees which can be exponentially big in size, while for PSG and DSG we only use $O(H)$ spanning graphs.

Because PSG and DSG approaches come close to each other, we additionally test for each value of $\alpha$ on each plot, whether the difference $err_{PSG} - err_{DSG}$ is bigger than zero. We apply a one-sided Wilcoxon's test \citep{wilcoxon} together with the Bonferroni correction because we test $33$ times \citep{bonferroni}. In most settings, the improvement is statistically significant (Figure \ref{fig:plot}). 


\section{Conclusion} \label{sec:concl}

In this manuscript, we, first of all, describe an algorithm for $O(N^\frac32)$ inference and sampling of planar zero-field Ising models on $N$ spins. Then we introduce a new family of zero-field Ising models composed of planar components and graphs of $O(1)$ size. For these models, we describe a polynomial algorithm for exact inference and sampling provided that the decomposition tree is also in the input. 
A theoretical application is $O(N^\frac32)$ inference and sampling algorithm for $K_{33}$-free or $K_5$-free zero-field Ising models--- both families are supersets of the family of planar zero-field models, and they are both neither treewidth- nor genus-bounded. We show that our scheme offers an improvement of the approximate inference scheme for arbitrary topologies.  The suggested improvement is based on the planar spanning graph ideas from \citet{globerson} but we use tractable spanning decomposition-based graphs instead of planar graphs. (That is we keep the algorithm of \citet{globerson}, but substitute planar graphs with a family of spanning decomposition-based graphs that are tractable.) This improvement of \citet{globerson} results in a tighter upper bound on the true partition function and a more precise approximation of marginal probabilities.

\begin{figure}
    \centering
    \includegraphics[width=0.9\linewidth]{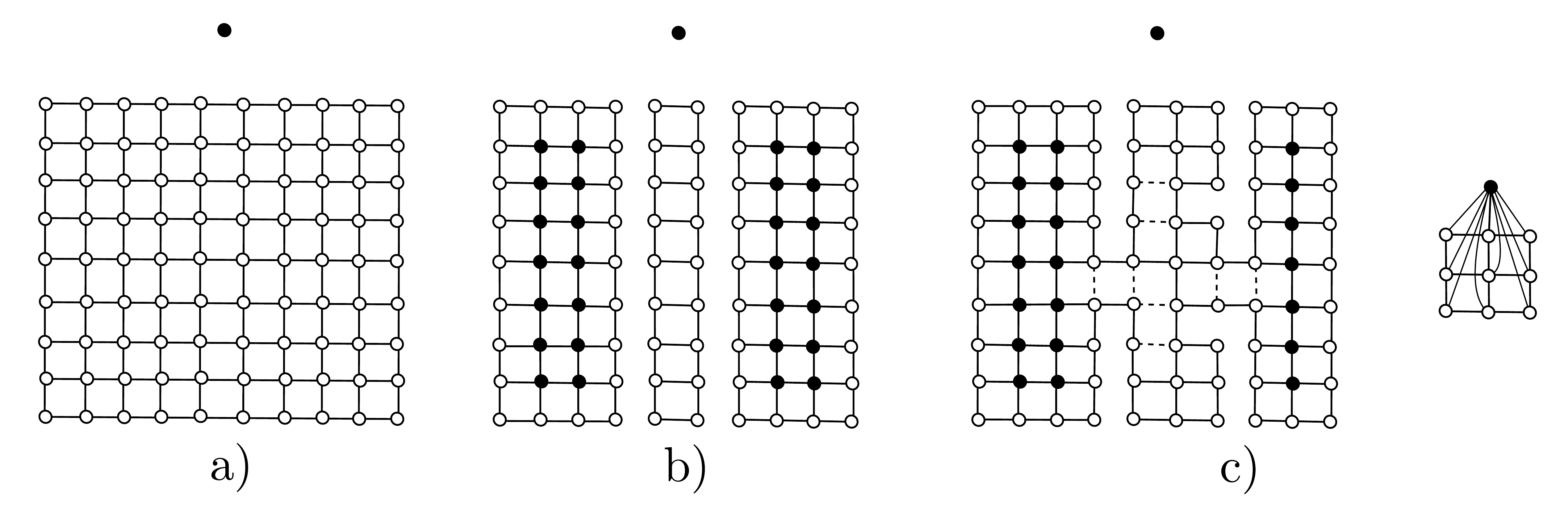}
    \caption{Construction of graphs used for approximate inference on a rectangular lattice. For better visualization, vertices connected to an apex are colored white. a) $G'$ graph. b) One of planar $G^{(r)}$ graphs used in \protect\citet{globerson}. Such ``separator'' pattern is repeated for each column and row, resulting in $2(H - 1)$ graphs in $\{ G^{(r)} \}$. In addition, \protect\citet{globerson} adds an \textit{independent variables} graph where only apex edges are drawn. c) A modified ``separator'' pattern we propose. Again, the pattern is repeated horizontally and vertically resulting in $2(H - 2)$ graphs $+$ independent variables graph. This pattern covers more magnetic fields and connects separated parts. Dashed edges indicate the structure of $10$-nice decomposition used for inference. (Nonplanar node of size $10$ is illustrated on the right.)}
    \label{fig:grid}
\end{figure}

\begin{figure}
    \centering
    \includegraphics[width=0.9\linewidth]{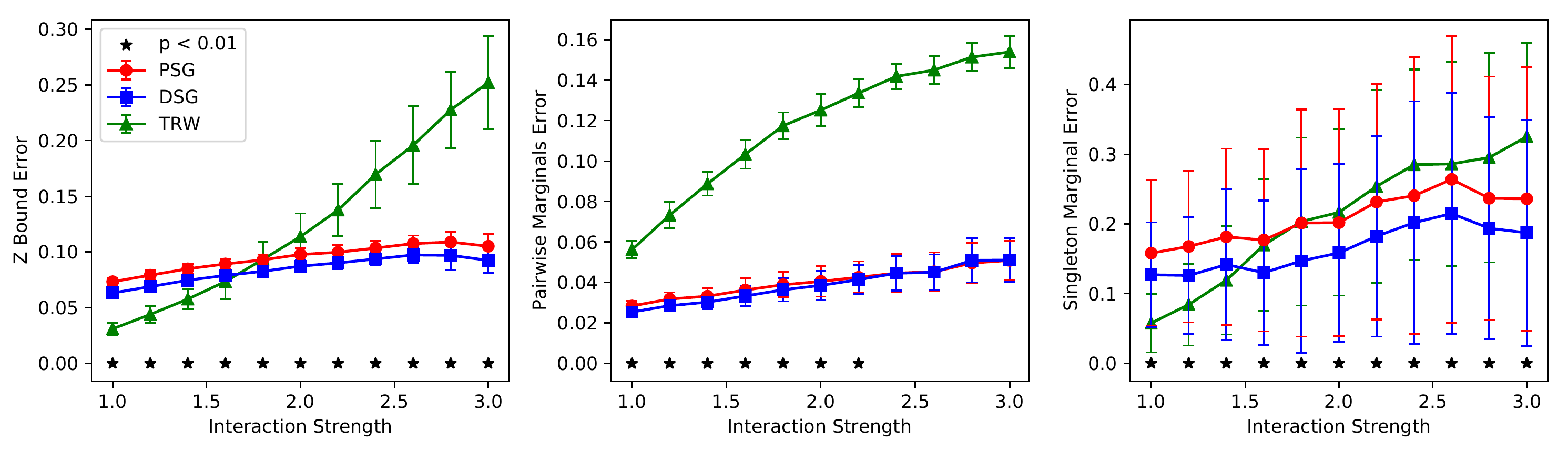}
    \caption{Comparison of tree-reweighted approximation (TRW), planar spanning graph (PSG), and decomposition-based spanning graph (DSG) approaches. The first plot is for normalized log-partition error, the second is for error in pairwise marginals, and the third is for error in singleton central marginal. Standard errors over $100$ trials are shown as error bars. An asterisk ``*'' indicates the statistically significant improvement of DSG over PSG, with a p-value smaller than $0.01$ according to the Wilcoxon test with the Bonferroni correction 
    \protect\citep{wilcoxon}. 
    }
    \label{fig:plot}
\end{figure}



\newpage

\appendix

\section{Lemma Proofs} \label{sec:lp}

\subsection{Lemma \ref{lemma:bij}}

\begin{proof} Let $E' \in \text{PM}(G^*)$. Call $e \in E$ \textit{saturated}, if it intersects an edge from $E' \cap E^*_I$. Each Fisher city is incident to an odd number of edges in $E' \cap E^*_I$. Thus, each face of $G$ has an even number of unsaturated edges. This property is preserved, when two faces/cycles are merged into one by evaluating respective symmetric difference. Therefore, one gets that any cycle in $G$ has an even number of unsaturated edges.

For each $i$ define $x_i := -1^{r_i}$, where $r_i$ is the number of unsaturated edges on the path connecting $v_1$ and $v_i$. The definition is consistent due to aforementioned cycle property. Now for each $e = \{ v, w \} \in E(G)$, $x_v = x_w$ if and only if $e$ is saturated. To conclude, we constructed $X$ such that $E' = M(X)$. Such $X$ is unique, because parity of unsaturated edges on a path between $v_1$ and $v_i$ uniquely determines relationship between $x_1$ and $x_i$, and $x_1$ is always $+1$.
\end{proof}

\subsection{Lemma \ref{lemma:zfitopm}}

\begin{proof} Let  $X' = (x'_1, ..., x'_N) \in \mathcal{C}_+$, $M(X') = E'$. The statement is justified by the following chain of transitions:
\begin{align}
\mathbb{P} ( M(S) = E' )
&= \mathbb{P} ( S = X') + \mathbb{P} ( S = -X' ) \nonumber \\
&= \frac2Z \exp \left(\sum_{e = \{ v, w \} \in E(G)} J_e x'_v x'_w \right) \nonumber \\
&= \frac2Z \exp \left(\sum_{e^* \in E' \cap E^*_I} 2 J_{g(e^*)} - \sum_{e \in E(G)} J_e\right) \nonumber \\
&= \frac2Z \exp \left(- \sum_{e \in E(G)} J_e\right) \prod_{e^* \in E' \cap E^*_I} c_{e^*} \nonumber \\
&= \frac2Z \exp \left(- \sum_{e \in E(G)} J_e\right) \prod_{e^* \in E'} c_{e^*} \nonumber \\
&= \frac{1}{Z^*} \prod_{e^* \in E'} c_{e^*}
\end{align}
\end{proof}

\subsection{Lemma \ref{lemma:cond}}

\begin{proof} We consider cases depending on $\omega$ and consequently reduce each case to a simpler one. For convenience in cases where applies we denote $u \triangleq v^{(1)}, h \triangleq v^{(2)}, q \triangleq v^{(3)}$:

\begin{enumerate}

\item \textbf{Conditioning on $\omega = 0$ spins.} Trivial given the algorithm described in section \ref{ch:planar}.

\item \textbf{Conditioning on $\omega = 1$ spin.} Since configurations $X$ and $-X$ have the same probability in $\mathcal{I}$, one deduces that $Z_{\,| \, x_u = s^{(1)}} = \frac12 Z$.

One also deduces that sampling $X$ from $\mathbb{P}(X\, | \, x_u = s^{(1)})$ is reduced to 1) drawing $\overline{X} = \{ \overline{x}_v = \pm 1 \}$ from $\mathbb{P}(X)$ and then 2) returning $X = (s^{(1)} \overline{x}_u) \cdot \overline{X}$ as a result.

\item \textbf{Conditioning on $\omega = 2$ spins.} There is an edge $e^{0} = \{ u, h \} \in E(G)$. The following expansion holds:
\begin{align}
    Z_{\,| \, x_u = s^{(1)}, x_h = s^{(2)}} &= \sum_{X, \, x_u = s^{(1)}, \, x_h = s^{(2)}} \exp\bigl( \sum_{ e = \{ v, w \} \in E(G)} J_e x_v x_w \bigr) \nonumber \\
    & = \exp (J_{e^0} s^{(1)} s^{(2)}) \cdot \sum_{X, \, x_u = s^{(1)}, \, x_h = s^{(2)}} \exp\bigl( \sum_{\substack{e = \{ v, w \} \in E(G) \\ e \neq e^0}} J_e x_v x_w\bigr) \nonumber \\
    &= \exp (J_{e^0} s^{(1)} s^{(2)}) \cdot \sum_{X, \, x_u = s^{(1)}, \, x_h = s^{(2)}} \exp\bigl( \sum_{\substack{e = \{ v, w \} \in E(G) \\ e \cap e^0 = \varnothing}} J_e x_v x_w \nonumber \\
    & + \sum_{\substack{e = \{ u, v \} \in E(G) \\ v \neq h}} (J_e s^{(1)}) x_v \cdot 1 + \sum_{\substack{e = \{ h, v \} \in E(G) \\ v \neq u}} (J_e s^{(2)}) x_v \cdot 1\bigr) \label{eq:cond2v}
\end{align}

Obtain graph $G'$ from $G$ by contracting $u, h$ into $z$. $G'$ is still planar and has $N - 1$ vertices. Preserve pairwise interactions of edges which were not deleted after contraction. For each edge $e = \{u, v\}$, $v \neq h$ set $J_{\{ z, v \}} = J_e s^{(1)}$, for each edge $e = \{ h, v \}$, $v \neq u$ set $J_{\{ z, v \}} = J_e s^{(2)}$. Collapse double edges in $G'$ which were possibly created by transforming into single edges. A pairwise interaction of the result edge is set to the sum of collapsed interactions.

Define a zero-field Ising model $\mathcal{I}'$ on the resulted graph $G'$ with its pairwise interactions, inducing a distribution $\mathbb{P}' (X' = \{ x'_v = \pm 1 | v \in V(G') \})$. Let $Z'$ denote $\mathcal{I}'$'s partition function. A closer look at (\ref{eq:cond2v}) reveals that
\begin{equation}
    Z_{\, | \, x_u = s^{(1)}, x_h = s^{(2)}} = \exp (J_{e^0} s^{(1)} s^{(2)}) \cdot Z'_{\,|\, x'_z = 1}
    \label{eq:ztrans}
\end{equation}
where $Z'_{\,|\, z'_y = 1}$ is a partition function conditioned on a single spin and can be found efficiently as shown above.

Since the equality of sums (\ref{eq:ztrans}) holds summand-wise, for a given $X'' = \{x''_v = \pm 1 \, | \, v \in V(G) \setminus \{ u, h \} \}$ the probabilities $\mathbb{P}(X'' \cup \{ x_u = s^{(1)}, x_h = s^{(2)} \} \, | \, x_u = s^{(1)}, x_h = s^{(2)})$ and $\mathbb{P}'(X'' \cup \{ x'_z = 1 \} \, | \, x'_z = 1)$ are the same. Hence, sampling from $\mathbb{P}(X \, | \, x_u = s^{(1)}, x_h = s^{(2)})$ is reduced to conditional sampling from planar zero-field Ising model $\mathbb{P}'(X' \, | \, x'_z = 1)$ of size~$N - 1$.

\item \textbf{Conditioning on $w = 3$ spins.} Without loss of generality assume that $u, h$ are connected by an edge $e^0$ in $G$. A derivation similar to (\ref{eq:cond2v}) and (\ref{eq:ztrans}) reveals that (preserving the notation of Case 2)
\begin{equation}
    Z_{\, | \, x_u = s^{(1)}, x_h = s^{(2)}, x_q = s^{(3)}} = \exp (J_{e^0} s^{(1)} s^{(2)}) \cdot Z'_{\,|\, x'_z = 1, x'_q = s^{(3)}}
\end{equation}
which reduces inference conditional on $3$ vertices to a simpler case of 2 vertices. Again, sampling from $\mathbb{P}(X \, | \, x_u = s^{(1)}, x_t = s^{(2)}, x_q = s^{(3)})$ is reduced to a more basic sampling from $\mathbb{P}'(X' \, | \, x'_z = 1, x'_q = s^{(3)})$.
\end{enumerate}
\end{proof}
In principle, Lemma \ref{lemma:cond} can be extended to arbitrarily large $\omega$ leaving a certain freedom for the Ising model conditioning framework. However, in this manuscript we focus on a given special case which is enough for our goals.

\section{Theorem \ref{th:pmmodel} Proof} \label{ch:pl_proof}

\subsection{Counting PMs of Planar \texorpdfstring{$\hat{G}$}{G} in \texorpdfstring{$O(\hat{N}^\frac32)$}{O} time}

This section addresses inference part of Theorem \ref{th:pmmodel}.

\subsubsection{Pfaffian Orientation} \label{subseq:pf}

Consider an orientation on $\hat{G}$. $\hat{G}$'s cycle of even length (built on an even number of vertices) is said to be \textit{odd-oriented}, if, when all edges along the cycle are traversed in any direction, an odd number of edges are directed along the traversal. For $X \subseteq V(\hat{G})$ let $\hat{G} (X)$ denote a graph $( X, \{ e \in E(\hat{G}) | e \subseteq X \} )$. An orientation of $\hat{G}$ is called \textit{Pfaffian}, if all cycles $C$, such that $\text{PM}(\hat{G}(V(\hat{G}) - C)) \neq \varnothing$, are odd-oriented.

We will need $\hat{G}$ to contain a Pfaffian orientation, moreover the construction is easy.
\begin{theorem}
Pfaffian orientation of $\hat{G}$ can be constructed in $O(\hat{N})$.
\end{theorem}
\begin{proof}
This theorem is proven constructively, see e.g. \citet{wilson,vazirani}, or \citet{schraudolph-kamenetsky}, where the latter construction is based on specifics of the expanded dual graph.
\end{proof}

Construct a skew-symmetric sparse matrix $\mathcal{K} \in \mathbb{R}^{ \hat{N} \times \hat{N} }$ ($\to$ denotes orientation of edges):
\begin{equation}
\mathcal{K}_{ij} = \begin{cases} c_e & \text{if } \{ v_i, v_j \} \in E(\hat{G}), v_i \to v_j \\ -c_e & \text{if } \{ v_i, v_j \} \in E(\hat{G}), v_j \to v_i \\ 0 & \text{if } \{ v_i, v_j \} \notin E(\hat{G}) \end{cases}
\label{eq:kdef}
\end{equation}

The next result allows to compute PF $\hat{Z}$ of PM model on $\hat{G}$ in a polynomial time.
\begin{theorem}
$\det \mathcal{K} > 0$, $\hat{Z} = \sqrt{\det \mathcal{K}}$.
\end{theorem}
\begin{proof}
See, e.g., \citet{wilson} or \citet{kasteleyn}. 
\end{proof}

\subsubsection{Computing \texorpdfstring{$\det \mathcal{K}$}{K}}

LU-decomposition of a matrix $A = LU$, found via Gaussian elimination, where $L$ is a lower-triangular matrix with unit diagonals and $U$ is an upper-triangular matrix, would be a standard way of computing $\det A$, which is then equal to a product of the diagonal elements of $U$. However, this standard way of constructing the LU decomposition applies only if all $A$'s leading principal submatrices are nonsingular (See e.g. \citet{horn-johnson}, section 3.5, for detailed discussions). And already the $1\times 1$  leading principal submatrix of $\mathcal{K}$ is zero/singular.

Luckily, this difficulty can be resolved through the following construction. Take $\hat{G}$'s arbitrary perfect matching $E' \in \text{PM}(\hat{G})$. In the case of a general planar graph $E'$ can be found via e.g. Blum's algorithm \citep{blum} in $O(\sqrt{\hat{N}} | E(\hat{G}) |) = O(\hat{N}^\frac32)$ time, while for graph $G^*$ appearing in this paper $E'$ can be found in $O(N)$ from a spin configuration using $M$ mapping (e.g. $E' = E^*_I = M(\{ +1, ..., +1 \}) \in \text{PM}(G^*)$). Modify ordering of vertices, $V(\hat{G}) = \{ v_1, v_2, ..., v_{\hat{N}} \}$, so that $E' = \{ \{ v_1, v_2 \}, ..., \{ v_{\hat{N} - 1}, v_{\hat{N}} \} \}$. Build $\mathcal{K}$ according to the definition (\ref{eq:kdef}). Obtain $\overline{\mathcal{K}}$ from $\mathcal{K}$ by swapping column $1$ with column $2$, $3$ with $4$ and so on. This results in $\det \mathcal{K} = | \det \overline{\mathcal{K}} |$,  where the new $\overline{\mathcal{K}}$ is properly conditioned.
\begin{lemma}
$\overline{\mathcal{K}}$'s leading principal submatrices are nonsingular.
\end{lemma}
\begin{proof}
The proof, presented in \citet{wilson} for the case of unit weights~$c_e$, generalizes to arbitrary positive~$c_e$.
\end{proof}

Notice, that in the general case (of a matrix represented in terms of a general graph) complexity of the LU-decomposition is cubic in the size of the matrix. 
Fortunately, \textit{nested dissection} technique, discussed in the following subsection, allows to reduce complexity of computing $\hat{Z}$ to $O(\hat{N}^\frac32)$.

\subsubsection{Nested Dissection} \label{subsec:nd}

The partition $P_1, P_2, P_3$ of set $V(\hat{G})$ is a \textit{separation} of $\hat{G}$, if for any $v \in P_1, w \in P_2$ it holds that $\{ v, w \} \notin E(\hat{G})$. We refer to $P_1, P_2$ as the \textit{parts}, and to $P_3$ as the \textit{separator}.

Lipton and Tarjan (LT) \citep{lipton-tarjan} found an $O(\hat{N})$ algorithm, which finds a separation $P_1, P_2, P_3$ such that $\max ( | P_1 |, | P_2 | ) \leq \frac23 \hat{N}$ and $| P_3 | \leq 2^\frac32 \sqrt{ \hat{N} }$. The LT algorithm can be used to construct the so called \textit{nested dissection ordering} of $V(\hat{G})$. The ordering is built recursively, by first placing vertices of $P_1$, then $P_2$ and $P_3$, and finally permuting indices of $P_1$ and $P_2$ recursively according to the ordering of $\hat{G}(P_1)$ and $\hat{G}(P_2)$ (See \citet{lipton-rose-tarjan} for accurate description of details, definitions and analysis of the nested dissection ordering). As shown by \citet{lipton-rose-tarjan} the complexity of finding the nested dissection ordering is $O( \hat{N} \log \hat{N} )$.

Let $A$ be a $\hat{N} \times \hat{N}$ matrix with a \textit{sparsity pattern} of $\hat{G}$. That is, $A_{ij}$ can be nonzero only if $i = j$ or $\{ v_i, v_j \} \in \hat{E}$.
\begin{theorem}
\citep{lipton-rose-tarjan} If $\hat{V}$ is ordered according to the nested dissection and $A$'s leading principal submatrices are nonsingular, computing the LU-decomposition of $A$ becomes a problem of the $O(N^\frac32)$ complexity.
\end{theorem}
Notice, however, that we cannot directly apply the Theorem to $\overline{\mathcal{K}}$, because the sparsity pattern of $\mathcal{K}$ is asymmetric and does not correspond, in general, to any graph.

Let $G^{**}$ be a planar graph, obtained from $\hat{G}$, by contracting each edge in $E'$, $| V(G^{**}) | = | E' | = \frac12 \hat{N}$. Find and fix a nested dissection ordering  over $V(G^{**})$ (it takes $O(\hat{N} \log \hat{N})$ steps) and let the $\{ v_1, v_2 \}, \dots, \{ v_{\hat{N} - 1}, v_{\hat{N}} \}$ enumeration of $E'$ correspond to this ordering. Split $\mathcal{K}$ into  $2 \times 2$ cells and consider the sparsity pattern of the nonzero cells. One observes that the resulting sparsity pattern coincides with the sparsity patterns of $\overline{\mathcal{K}}$ and $G^{**}$.  Since LU-decomposition can be stated in the $2 \times 2$ block elimination form, its complexity is reduced down to $O(\hat{N}^\frac32)$.

This concludes construction of an efficient inference (counting) algorithm for planar PM model.

\subsection{Sampling PMs of Planar \texorpdfstring{$\hat{G}$}{G} in \texorpdfstring{$O(\hat{N}^\frac32)$}{O} time (Wilson's Algorithm)} \label{app:wilson}

This section addresses sampling part of Theorem \ref{th:pmmodel}. In this section we assume that degrees of $\hat{G}$'s vertices are upper-bounded by $3$.
This is true for $G^*$ - the only type of PM model appearing in the paper. Any other constant substituting $3$ wouldn't affect the analysis of complexity. Moreover, \citet{wilson} shows that any PM model on a planar graph can be reduced to bounded-degree planar model without affecting $O(\hat{N}^\frac32)$ complexity.

\subsubsection{Structure of the Algorithm}

Denote a sampled PM as $M$, $\mathbb{P}(M) = \hat{Z}^{-1} \prod_{e \in M} c_e$. Wilson's algorithm first applies LT algorithm of \citet{lipton-tarjan} to find a separation $P_1, P_2, P_3$ of $\hat{G}$ ($\max ( | P_1 |, | P_2 | ) \leq \frac23 \hat{N}$, $| P_3 | \leq 2^\frac32 \sqrt{ \hat{N} }$). Then it iterates over $v \in P_3$ and for each $v$ it draws an edge of $M$, saturating $v$. Then it appears that, given this intermediate result, drawing remaining edges of $M$ may be split into two independent drawings over $\hat{G}(P_1)$ and $\hat{G}(P_2)$, respectively, and then the process is repeated recursively. 

It takes $O(\hat{N}^\frac32)$ steps to sample edges attached to $P_3$ at the first step of the recursion, therefore the overall complexity of the Wilson's algorithm is also $O(\hat{N}^\frac32)$.

Subsection \ref{subsec:probs} introduces probabilities required to draw the aforementioned PM samples. Subsections \ref{subsec:step1} and \ref{subsec:step2} describe how to sample edges attached to the separator, while Subsection \ref{subsec:step3} focuses on describing the recursion.

\subsubsection{Drawing Perfect Matchings} \label{subsec:probs}

For some $Q \in E(\hat{G})$ consider the probability of getting $Q$ as a subset of $M$:
\begin{align}
    \mathbb{P}(Q \subseteq M) &= \frac{1}{\hat{Z}} \sum_{\substack{M' \in \text{PM}(\hat{G}) \\ Q \subseteq M'}} \biggl( \prod_{e \in M'} c_e \biggr) \nonumber \\
    &= \frac{1}{\hat{Z}} \biggl( \prod_{e \in Q} c_e \biggr) \cdot \sum_{M' \in \text{PM}(\hat{G})} \biggl( \prod_{e \in M' \setminus Q} c_e \biggr)
\label{eq:qprob}
\end{align}

Let $\hat{V}_Q = \cup_{e \in Q} e$ and $\hat{G}_{\setminus Q} = \hat{G} (V(\hat{G}) \setminus \hat{V}_Q)$. Then the set $\{ M' \setminus Q \, | \, M' \in \text{PM}(\hat{G}) \}$ coincides with $\text{PM}(\hat{G}_{\setminus Q})$. This yields the following expression
\begin{equation}
    \mathbb{P}(Q \subseteq M) = \frac{\hat{Z}_{\setminus Q}}{\hat{Z}} \biggl( \prod_{e \in Q} c_e \biggr)
\end{equation}
where
\begin{equation}
    \hat{Z}_{\setminus Q} = \sum_{M'' \in \text{PM}(\hat{G}_{\setminus Q})} \biggl( \prod_{e \in M''} c_e \biggr)
\end{equation}
is a PF of the PM model on $\hat{G}_{\setminus Q}$ induced by the edge weights $c_e$.

For a square matrix $A$ let $A_{c_1, ..., c_l}^{r_1, ..., r_l}$ denote the matrix obtained by deleting rows $r_1, ..., r_l$ and columns $c_1, ..., c_l$ from $A$. Let $[A]_{c_1, ..., c_l}^{r_1, ..., r_l}$ be obtained by leaving only rows $r_1, ..., r_l$ and columns $c_1, ..., c_l$ of $A$ and placing them in this order.

Now let $\hat{V}_Q = \{ v_{i_1}, ..., v_{i_r} \}, i_1 < ... < i_r$. A simple check demonstrates that deleting vertex from a graph preserves the Pfaffian orientation. By induction this holds for any number of vertices deleted. From that it follows that $\mathcal{K}_{i_1, ..., i_r}^{i_1, ..., i_r}$ is a Kasteleyn matrix for $\hat{G}_{\setminus Q}$ and then
\begin{equation}
    \hat{Z}_{\setminus Q} = \pfaffian \mathcal{K}_{i_1, ..., i_r}^{i_1, ..., i_r} =  \sqrt{\det \mathcal{K}_{i_1, ..., i_r}^{i_1, ..., i_r}}
\end{equation}
resulting in
\begin{equation}
    \mathbb{P}(Q \subseteq M) = \sqrt{ \frac{\det \mathcal{K}_{i_1, ..., i_r}^{i_1, ..., i_r}}{\det \mathcal{K}} } \cdot \biggl( \prod_{e \in Q} c_e \biggr)
\end{equation}

Linear algebra transformations, described by \citet{wilson}, suggest that if $A$ is non-singular, then
\begin{equation}
    \frac{\det A_{c_1, ..., c_l}^{r_1, ..., r_l}}{\det A} = \pm \det [A^{-1}]_{r_1, ..., r_l}^{c_1, ..., c_l}
\end{equation}
This observation allows us to express probability (\ref{eq:qprob}) as
\begin{equation}
    \mathbb{P}(Q \subseteq M) = \sqrt{ | \det [\mathcal{K}^{-1}]_{i_1, ..., i_r}^{i_1, ..., i_r} |} \cdot \biggl( \prod_{e \in Q} c_e \biggr)
\end{equation}

Now we are in the position to describe the first step of the Wilson's recursion.

\subsubsection{Step 1: Computing Lower-Right Submatrix of \texorpdfstring{$\overline{\mathcal{K}}^{-1}$}{K}} \label{subsec:step1}

Find a separation $P_1, P_2, P_3$ of $\hat{G}$. The goal is to sample an edge from every $v \in P_3$.

Let $W$ be a set of vertices from $P_3$ and their neighbors, then $| W | \leq 3 | P_3 |$ because each vertex in $\hat{G}$ is of degree at most $3$. Let $W^{**} \subseteq V(G^{**})$ be a set of the contracted edges (recall $G^{**}$ definition from Subsection \ref{subsec:nd}), containing at least one vertex from $W$, $| W^{**} | \leq | W |$. Then $W^{**}$ is a separator of $G^{**}$ such that
\begin{equation} \label{eq:tstarstar}
    | W^{**} | \leq | W | \leq 3 | P_3 | \leq 3 \cdot 2^\frac32 \sqrt{\hat{N}} \leq 3 \cdot 2^2 \sqrt{| V(G^{**}) |}
\end{equation}
where one uses that, $| V(G^{**}) | = \frac{\hat{N}}{2}$.
Find a nested dissection ordering (Subsection \ref{subsec:nd}) of $V(G^{**})$ with $W^{**}$ as a top-level separator. This is a correct nested dissection due to Eq.~(\ref{eq:tstarstar}).

Utilizing this ordering, construct $\overline{\mathcal{K}}$. Compute $L$ and $U$ - LU-decomposition of $\overline{\mathcal{K}}$ ($O(\hat{N}^\frac32)$ time).
Let $\gamma = 2 |W^{**}| \leq 3 \cdot 2^\frac52 \sqrt{\hat{N}}$ and let $\mathcal{I}$ be a shorthand notation for $(\hat{N} - \gamma + 1, ..., \hat{N})$. Using $L$ and $U$, find $D = [ \overline{\mathcal{K}}^{-1} ]_\mathcal{I}^\mathcal{I}$, which is a lower-right $\overline{\mathcal{K}}^{-1}$'s submatrix of size $\gamma \times \gamma$.

It is straightforward to observe that the $i$-th column of $D$, $d_i$, satisfies
\begin{equation}
    [L]_\mathcal{I}^\mathcal{I} \times \biggl( [U]_\mathcal{I}^\mathcal{I} \times d_i \biggr) = e_i,
\end{equation}
where $e_i$ is a zero vector with unity at the $i$-th position. Therefore constructing $D$ is reduced to solving $2 \gamma$ triangular systems, each  of size $\gamma \times \gamma$, resulting in $O(\gamma^3) = O(\hat{N}^\frac32)$ required steps.

\subsubsection{Step 2: Sampling Edges in the Separator} \label{subsec:step2}

Now, progressing iteratively, one finds $v \in P_3$ which is not yet paired and draw an edge emanating from it. Suppose that the edges, $e_1 = \{ v_{j_1}, v_{j_2} \}, ..., e_k = \{ v_{j_{2 k - 1}}, v_{j_{2 k}} \}$, are already sampled. We assume that by this point we have also computed LU-decomposition $A_k = [\mathcal{K}^{-1}]_{j_1, ..., j_{2 k}}^{j_1, ..., j_{2 k}} = L_k U_k$ and we will update it to $A_{k + 1}$ when the new edge is drawn. Then
\begin{equation}
    \mathbb{P}(e_1, ..., e_k \in M) = \sqrt{| \det A_k |} \prod_{j = 1}^k c_{e_j}
\end{equation}

Next we choose $j_{2 k + 1}$ so that $v_{j_{2 k + 1}}$ is not saturated yet. We iterate over $v_{j_{2 k + 1}}$'s neighbors considered as candidates for becoming $v_{j_{2 k + 2}}$. Let $v_j$ to become the next candidate, denote $e_{k + 1} = \{ v_{j_{2 k + 1}}, v_j \}$. For $n \in \mathbb{N}$ let $\alpha(n) = n + 1$ if $n$ is odd and $\alpha(n) = n - 1$ if $n$ is even. Then the identity
\begin{equation}
    \mathcal{K}^{-1} = [\overline{\mathcal{K}}^{-1}]_{1, 2, ..., \hat{N}}^{\alpha(1), \alpha(2), ..., \alpha(\hat{N})},
    \label{eq:app_K^-1}
\end{equation}
follows from the definition of $\overline{\mathcal{K}}$. One deduces from Eq.~(\ref{eq:app_K^-1})
\begin{equation}
    A_{k + 1} = [\mathcal{K}^{-1}]_{j_1, ..., j_{2 k + 1}, j}^{j_1, ..., j_{2 k + 1}, j} = [\overline{\mathcal{K}}^{-1}]_{j_1, ..., j_{2 k + 1}, j}^{\alpha (j_1), ..., \alpha(j_{2 k + 1}), \alpha(j)}
\end{equation}

Constructing $W^{**}$ one has $j_1, ..., j_{2 k + 1}, j, \alpha(j_1), ..., \alpha(j_{2 k + 1}), \alpha(j) > \hat{N} - t$. It means that $A_{k + 1}$ is a submatrix of $D$ with permuted rows and columns, hence $A_{k + 1}$ is known.

We further observe that
\begin{equation}
    A_{k + 1} = \begin{bmatrix} A_k & y \\ r & d \end{bmatrix} = \begin{bmatrix} L_k & 0 \\ R & 1 \end{bmatrix} \begin{bmatrix} U_k & Y \\ 0 & z \end{bmatrix} = L_{k + 1} U_{k + 1}
\end{equation}
Therefore to update $L_{k + 1}$ and $U_{k + 1}$, one just solves the triangular system of equations $R U_k = r$ and $L_k Y = y$, where $R^\top, r^\top, Y, y$ are of size $2 k \times 2$ (this is done in $O(k^2)$ steps), and then compute $z = d - R Y$ which is of the size $2 \times 2$, then set, $u = \det z$.

The probability to pair $v_{j_{2 k + 1}}$ and $v_j$ is
\begin{align}
    \mathbb{P}(e_{k + 1} \in M \, | \, e_1, ..., e_k \in M) &= \frac{\mathbb{P}(e_1, ..., e_{k + 1} \in M)}{\mathbb{P}(e_1, ..., e_k \in M)} \nonumber \\
    &= \frac{ \sqrt{| \det A_{k + 1} |} \prod_{j = 1}^{k + 1} c_{e_j}}{ \sqrt{| \det A_k |} \prod_{j = 1}^k c_{e_j}} \nonumber \\
    &= \frac{ c_{e_{k + 1}} \sqrt{ | u | | \det A_k |} }{\sqrt{ | \det A_k | }} \nonumber \\
    &= c_{e_{k + 1}} \sqrt{| u |}
\end{align}

Therefore maintaining $U_{k + 1}$ allows us to compute the required probability and draw a new edge from $v_{j_{2 k + 1}}$. By construction of $\hat{G}$, $v_{j_{2 k + 1}}$ has only $3$ neighbors, therefore the complexity of this step is $O(\sum_{k = 1}^{| P_3 |} k^2) = O(\hat{N}^\frac32)$ because $| P_3 | \leq 2^\frac32 \sqrt{\hat{N}}$.

\subsection{Step 3: Recursion} \label{subsec:step3}

Let $M_{sep} = \{ e_1, e_2, ... \}$ be a set of edges drawn on the previous step, and $\hat{V}_{sep}$ be a set of vertices saturated by $M_{sep}$, $P_3 \subseteq \hat{V}_{sep}$. Given $M_{sep}$, the task of sampling $M \in \text{PM}(\hat{G})$ such that $M_{sep} \subseteq M$ is reduced to sampling perfect matchings $M_1$ and $M_2$ over $\hat{G}(P_1 \setminus \hat{V}_{sep})$ and $\hat{G}(P_2 \setminus \hat{V}_{sep})$, respectively. Then $M = M_1 \cup M_2 \cup M_{sep}$ becomes the result of the perfect matching drawn from (\ref{eq:pmprobs}).

Even though only the first step of the Wilson's recursion was discussed so far, any further step in the recursion is done in exactly the same way with the only exception that vertex degrees may become less than $3$, while in $\hat{G}$ they are exactly $3$. Obviously, this does not change the iterative procedure and it also does not affect the complexity analysis.

\section{Theorem \ref{th:k33dec} Proof} \label{ch:k33dec}

Prior to the proof we introduce a series of definitions and results. We follow \citet{hopcroft2,gutwenger}, see also \citet{08Mader} to define the tree of triconnected components. The definitions apply for a biconnected graph $G$ (see the definition of biconnected graph and biconnected component e.g. in Appendix \ref{sec:k5proof}.)

Let $v, w \in V(G)$. Divide $E(G)$ into equivalence classes $E_1, ..., E_k$ so that $e_1, e_2$ are in the same class if they lie on a common simple path that has $v, w$ as endpoints. $E_1, ..., E_k$ are referred to as \textit{separation classes}. If $k \geq 2$, then $\{ v, w \}$ is a \textit{separation pair} of $G$, unless (a) $k = 2$ and one of the classes is a single edge or (b) $k = 3$ and each class is a single edge. Graph $G$ is called \textit{triconnected} if it has no separation pairs.

\begin{figure*}
\centering
\includegraphics[width=0.9\linewidth]{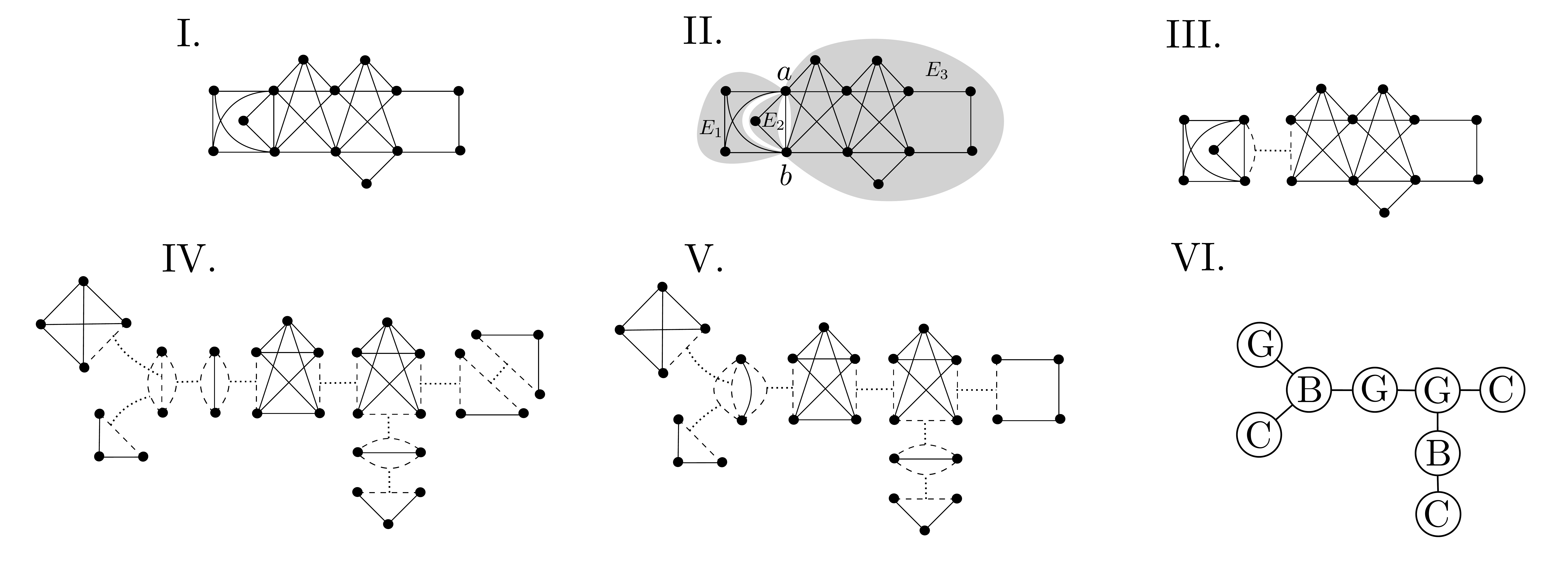}
\caption{(I) An example biconnected graph $G$. (II) A separation pair $\{ a, b \}$ of $G$ and separation classes $E_1, E_2, E_3$ associated with $\{ a, b \}$. (III) Result of split operation with $E' = E_1 \cup E_2, E'' = E_3$. Dashed lines indicate virtual edges and dotted lines connect equivalent virtual edges in split graphs. (IV) Split components of $G$ (non-unique). (V) Triconnected components of $G$. (VI) Triconnected component tree $T$ of $G$; spacial alignment of V is preserved. ``G," ``B," and ``C" are examples of the  ``triconnected graph,"  ``multiple bond," and  ``cycle," respectively.}
\label{fig:gr_tree}
\end{figure*}

Let $\{ v, w \}$ be a separation pair in $G$ with equivalence classes $E_1, ..., E_k$. Let $E' = \cup_{i = 1}^l E_l$, $E'' = \cup_{i = l + 1}^k E_l$ be such that $| E' | \geq 2$, $| E'' | \geq 2$. Then, graphs $G_1 = (\cup_{e \in E'} e, E' \cup \{ e_\mathcal{V} \}), G_2 = (\cup_{e \in E''} e, E'' \cup \{ e_\mathcal{V} \})$ are called \textit{split graphs} of $G$ with respect to $\{ v, w \}$, and $e_\mathcal{V}$ is a \textit{virtual edge}, which is a  new edge between $v$ and $w$, identifying the split operation. Due to the addition of $e_\mathcal{V}$, $G_1$ and $G_2$ are not normal in general.

Split $G$ into $G_1$ and $G_2$. Continue splitting $G_1, G_2$, and so on, recursively, until no further split operation is possible. The resulting graphs are \textit{split components} of $G$. They can either be $K_3$ (triangles), triple bonds, or triconnected normal graphs.

Let $e_\mathcal{V}$ be a virtual edge. There are exactly two split components containing $e_\mathcal{V}$: $G_1 = (V_1, E_1)$ and $G_2 = (V_2, E_2)$. Replacing $G_1$ and $G_2$ with $G' = (V_1 \cup V_2, (E_1 \cup E_2) \setminus \{ e_\mathcal{V} \})$ is called \textit{merging} $G_1$ and $G_2$. Do all possible mergings of the cycle graphs (starting from triangles), and then do all possible mergings of multiple bonds starting from triple bonds. Components of the resulting set are referred to as the \textit{triconnected components} of $G$. We emphasize again that some graphs (i.e., cycles and bonds) in the set of triconnected components are not necessarily triconnected.

\begin{lemma}
\citep{hopcroft2} Triconnected components are unique for $G$. Total number of edges within the triconnected components is at most $3 | E | - 6$.
\label{lemma:3}
\end{lemma}

Consider a graph $T'$, where vertices (further referred to as \textit{nodes} for disambiguation) are triconnected components, and there is an edge between $a$ and $b$ in $T'$, when $a$ and $b$ share a (copied) virtual edge.

\begin{lemma}
\citep{hopcroft2} $T'$ is a tree.
\end{lemma}

We will also use the following celebrated result:
\begin{lemma} \label{th:k33free}
\citep{hall} Biconnected graph $G$ is $K_{33}$-free if and only if its nonplanar triconnected components are exactly~$K_5$.
\end{lemma}

The graph on Figure \ref{fig:gr_tree} is actually $K_{33}$-free according to the Lemma. Now we are in the position to give a proof of the Theorem \ref{th:k33dec}.

\begin{proof}
Since $G$ is $K_{33}$-free and has no loops or multiple edges, it holds that $| E(G) | = O(N)$ \citep{thomason}. In time $O(N)$ we can find a forest of $G$'s biconnected components \citep{tarjan}. If we find the $5$-nice decomposition of each biconnected component, we can trivially combine them into a single $5$-nice decomposition in time $O(N)$ using navels of size $0$ and $1$. Hence, we can assume that $G$ is biconnected.

Build a tree of triconnected components for $G$ in time $O(N)$ \citep{hopcroft2,gutwenger,vo}. Now delete virtual edges, which results in a $5$-nice decomposition of $G$, given the Lemma \ref{th:k33free}.
\end{proof}

\section{
Proof for Theorem \ref{th:k5dec}} \label{sec:k5proof}

Prior to the proof, we introduce a series of definitions used by \citet{reed}. It is assumed that a graph $G = (V, E)$ (no loops and multiple edges) is given.

For any $X \subseteq V(G)$ let $G - X$ denote a graph $(V(G) \setminus X, \{ e = \{ v, w \} \in E(G) \, | \, v, w \notin X \})$. $X \subseteq V(G)$ is a $(i, j)$-\textit{cut} whenever $| X | = i$ and $G - X$ has at least $j$ connected components.

The graph is \textit{biconnected} whenever it has no $(1, 2)$-cut. A \textit{biconnected component} of the graph is a maximal biconnected subgraph. Clearly, a pair of biconnected components can intersect in at most one vertex and a graph of components' intersections is a tree when $G$ is connected (\textit{a tree of biconnected components}). The graph is \textit{$3$-connected} whenever it has no $(2, 2)$-cut.

A \textit{$2$-block tree} of a biconnected graph $G$, written $\langle T', \mathcal{G}' \rangle$, is a tree $T'$ with a set $\mathcal{G}' = \{ G'_t \}_{t \in V(T')}$ with the following properties:
\begin{itemize}
    \item[--] $G'_t$ is a graph (possibly with multiple edges) for each $t \in V(T')$.
    \item[--] If $G$ is $3$-connected then $T'$ has a single node $r$ which is colored $1$ and $G'_r = G$.
    \item[--] If $G$ is not $3$-connected then there exists a color $2$ node $t \in V(T')$ such that
    \begin{enumerate}
        \item $G'_t$ is a graph with two vertices $u$ and $v$ and no edges for some $(2, 2)$-cut $\{ u, v \}$ in $G$.
        \item Let $T'_1, \dots, T'_k$ be the connected components (subtrees) of $T' - t$. Then $G - \{ u, v \}$ has $k$ connected components $U_1, \dots, U_k$ and there is a labelling of these components such that $T'_i$ is a $2$-block tree of $G'_i = (V(U_i) \cup \{ u, v \}, E(U_i) \cup \{ \{ u, v \} \})$.
        \item For each $i$, there exists exactly one color $1$ node $t_i \in V(T'_i)$ such that $\{ u, v \} \subseteq V(G'_{t_i})$.
        \item For each $i$, $\{ t, t_i \} \in E(T)$.
    \end{enumerate}
\end{itemize}

A \textit{$(3, 3)$-block tree} of a $3$-connected graph $G$, written $\langle T'', \mathcal{G}'' \rangle$, is a tree $T''$ with a set $\mathcal{G}'' = \{ G''_t \}_{t \in V(T'')}$ with the following properties:
\begin{itemize}
    \item[--] $G''_t$ is a graph (possibly with multiple edges) for each $t \in V(T'')$.
    \item[--] If $G$ has no $(3, 3)$-cut then $T$ has a single node $r$ which is colored $1$ and $G_r = G$.
    \item[--] If $G$ has a $(3, 3)$-cut then there exists a color $2$ node $t \in V(T'')$ such that
    \begin{enumerate}
        \item $G''_t$ is a graph with vertices $u$, $v$ and $w$ and no edges for some $(3, 3)$-cut $\{ u, v, w \}$ in $G$.
        \item Let $T''_1, \dots, T''_k$ be the connected components (subtrees) of $T'' - t$. Then $G - \{ u, v, w \}$ has $k$ connected components $U_1, \dots, U_k$ and there is a labelling of these components such that $T_i$ is a $(3, 3)$-block tree of $G''_i = (V(U_i) \cup \{ u, v, w \}, E(U_i) \cup \{ \{ u, v \}, \{ v, w \}, \{ u, w \} \})$.
        \item For each $i$, there exists exactly one color $1$ node $t_i \in V(T''_i)$, such that $\{ u, v, w \} \subseteq V(G''_{t_i})$.
        \item For each $i$, $\{ t, t_i \} \in E(T'')$.
    \end{enumerate}
\end{itemize}

\begin{proof}
Since $G$ is $K_5$-free and has no loops or multiple edges, it holds that $| E(G) | = O(N)$ \citep{thomason}. In time $O(N)$ we can find a forest of $G$'s biconnected components \citep{tarjan}. If we find an $8$-nice decomposition for each biconnected component, join them into a single $8$-nice decomposition by using attachment sets of size $1$ for decompositions inside $G$'s connected component and attachment sets of size $0$ for decompositions in different connected components. Hence, further we assume that $G$ is biconnected.

The $O(N)$ algorithm of \citet{reed} finds a $2$-block tree $\langle T', \mathcal{G}' \rangle$ for $G$ and then for each color $1$ node $G'_t \in \mathcal{G}'$ it finds $(3, 3)$-block tree $\langle T'', \mathcal{G}'' \rangle$ where all components are either planar or M\"{o}bius ladders. To get an $8$-nice decomposition from each $(3, 3)$-block tree, 1) for each color $2$ node contract an edge between it and one of its neighbours in $T''$ and 2) remove all edges which were only created during $\langle T'', \mathcal{G}'' \rangle$ construction (2nd item of $(3, 3)$-block tree definition).

Now we have to draw additional edges in the forest $F$ of obtained $8$-nice decompositions so that to get a single $8$-nice decomposition $\mathcal{T}$ of $G$. Notice that for each pair of adjacent nodes $G'_t, G'_s \in \mathcal{G}'$ where $G'_t$ is color $1$ node and $G'_s = (\{ u, v \}, \varnothing)$ is a color $2$ node, $u, v$ are in $V(G'_t)$ and $\{ u, v \} \in E(G'_t)$. Hence, there is at least one component $G''_{r}$ of $8$-nice decomposition of $G'_t$ where both $u$ and $v$ are present. For each pair of $s$ and $t$ draw an edge between $s$ and $r$ in $F$. Then 1) for each color $2$ node in $F$ (such as $s$) contract an edge between it and one of its neighbors (such as $r$) and 2) remove all edges which were created during $\langle T', \mathcal{G}' \rangle$ construction (2nd item of $2$-block tree definition). This results is a correct $c$-nice decomposition for biconnected $G$.
\end{proof}

\section{Random Graph Generation} \label{sec:grgen}

As our derivations cover the most general case of planar and $K_{33}$-free graphs, we want to test them on graphs which are as general as possible. Based on Lemma \ref{th:k33free} (notice, that it provides necessary and sufficient conditions for a graph to be $K_{33}$-free) we implement a randomized construction of $K_{33}$-free graphs, which is assumed to cover most general $K_{33}$-free topologies.

Namely, one generates a set of $K_5$'s and random planar graphs, attaching them by edges to a tree-like structure.
Our generation process consists of the following two steps.
\begin{enumerate}
\item \textbf{Planar graph generation.} This step accepts $N \geq 3$ as an input and generates a normal biconnected planar graph of size $N$ along with its embedding on a plane. The details of the construction are as follows.

First, a random embedded tree is drawn iteratively. We start with a single vertex, on each iteration choose a random vertex of an already ``grown'' tree, and add a new vertex connected only to the chosen vertex. Items I-V in Fig.~\ref{fig:planargen} illustrate this step.

Then we triangulate this tree by adding edges until the graph becomes biconnected and all faces are triangles, as in the Subsection \ref{subsec:edg} (VI in Figure \ref{fig:planargen}). Next, to get a normal graph, we remove multiple edges possibly produced by triangulation (VII in Fig.~\ref{fig:planargen}). At this point the generation process is complete.

\begin{figure}
\centering
\includegraphics[width=0.95\linewidth]{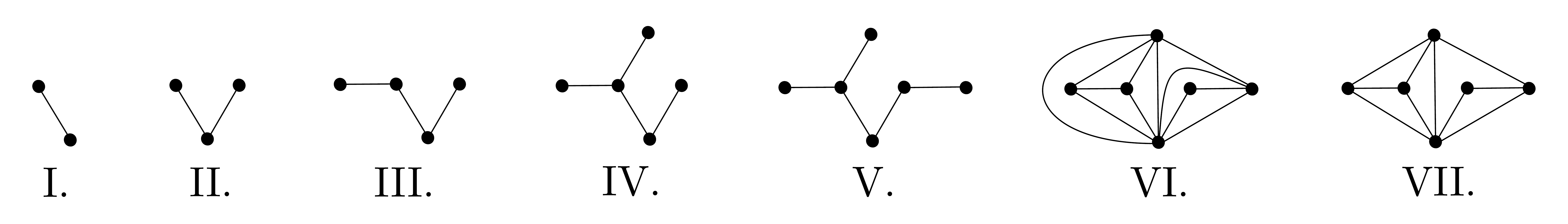}
\caption{Steps of planar graph generation. I-V refers to random tree construction on a plane, VI is a triangulation of a tree, VII is a result after multiple edges removal.}
\label{fig:planargen}
\end{figure}

\item \textbf{$K_{33}$-free graph generation.} Here we take $N \geq 5$ as the input and generate a normal biconnected $K_{33}$-free graph $G$ in a form of its partially merged decomposition $T$. Namely, we generate a tree $T$ of graphs where each node is either a normal biconnected planar graph or $K_5$, and every two adjacent graphs share a virtual edge. 

The construction is greedy and is essentially a tree generation process from Step 1. We start with $K_5$ root and then iteratively create and attach new nodes. Let $N' < N$ be a size of the already generated graph, $N' = 5$ at first. Notice, that when a node of size $n$ is generated, it contributes $n - 2$ new vertices to $G$.

An elementary step of iteration here is as follows. If $N - N' \geq 3$, a coin is flipped and the type of new node is chosen - $K_5$ or planar. If $N - N' < 3$, $K_5$ cannot be added, so a planar type is chosen. If a planar node is added, its size is drawn uniformly in the range between $3$ and $N - N' + 2$ and then the graph itself is drawn as described in Step 1. Then we attach a new node to a randomly chosen free edge of a randomly chosen node of $T'$. We repeat this process until $G$ is of the desired size $N$. Fig.~\ref{fig:k33freegen} illustrates the algorithm.
\end{enumerate}

To obtain an Ising model from $G$, we sample pairwise interactions for each edge of $G$ independently from $\mathcal{N}(0, 0.1^2)$.

Notice that the tractable Ising model generation procedure is designed in this section solely for the convenience of testing and it is not claimed to be sampling models of any particular practical interest (e.g. in statistical physics or computer science).

\begin{figure}
\centering
\includegraphics[width=0.69\linewidth]{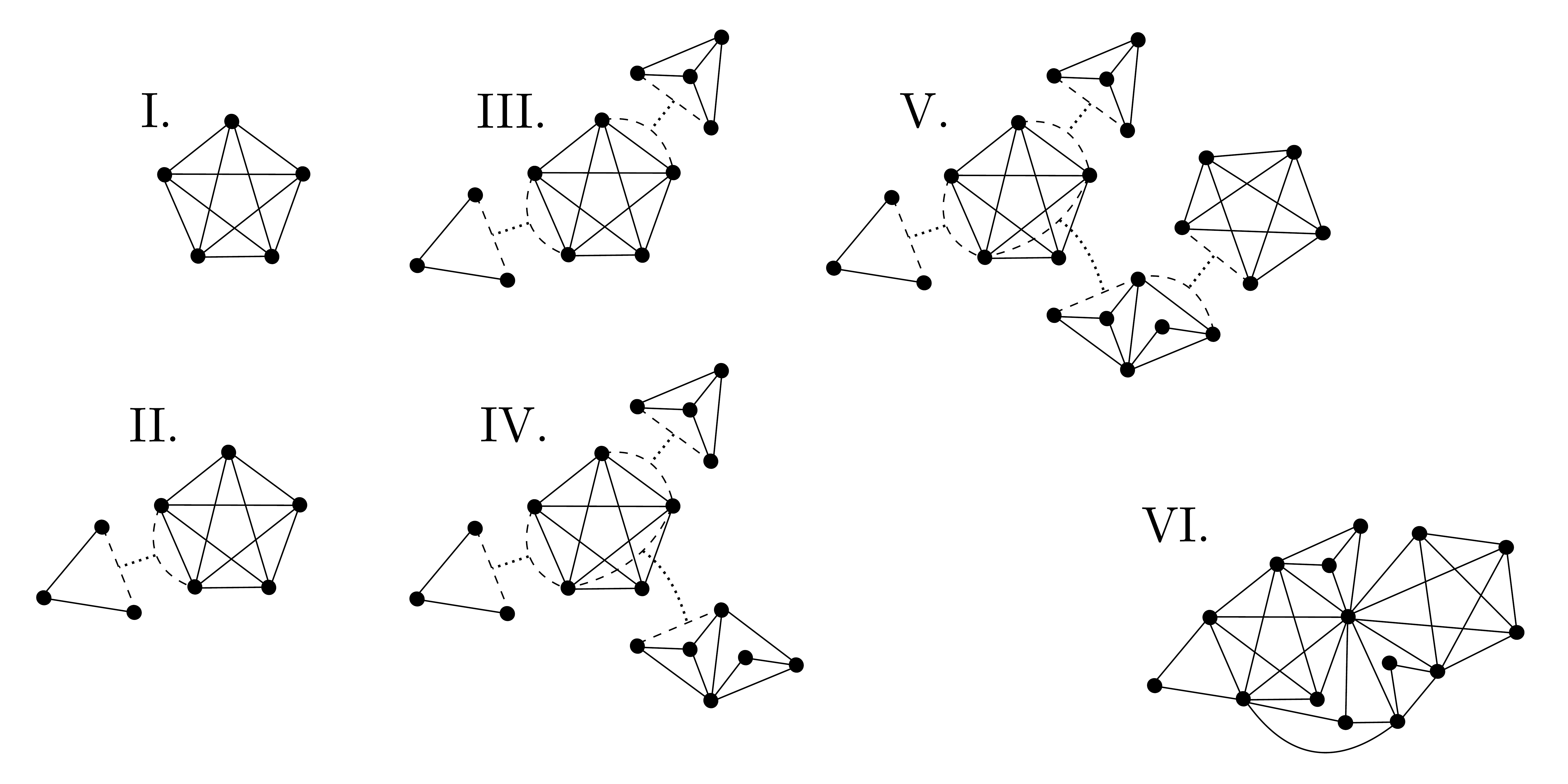}
\caption{Generation of $K_{33}$-free graph $G$ and its partially merged decomposition $T'$. Starting with $K_5$ (I), new components are generated and attached to random free edges (II-V). VI is a result graph $G$ obtained by merging all components in $T'$.}
\label{fig:k33freegen}
\end{figure}

\section{Upper Bound Minimization and Marginal Computation in Approximation Scheme} \label{sec:mcas}

Denote:
\begin{equation*}
h(J') \triangleq \min_{\rho(r) \geq 0, \sum_r \rho(r) = 1} g(J', \rho), \qquad g(J', \rho) \triangleq \min_{\{ J^{(r)} \}, \sum_r \rho(r) \hat{J}^{(r)} = J'} \sum_r \rho(r) \log Z(G^{(r)}, 0, J^{(r)})
\end{equation*}
where $h(J')$ is a tight upper bound for $\log Z (G', 0, J')$.

Given a fixed $\rho$, we compute $g(J', \rho)$ using L-BFGS-B optimization \citep{lbfgsb} by back-propagating through $Z(G^{(r)}, 0, J^{(r)})$ and projecting gradients on the constraint linear manifold. On the upper level we also apply L-BFGS-B algorithm to compute $h(J')$, which is possible since \citep{wainwright,globerson}
\begin{equation*}
    \frac{\partial}{\partial \rho(r)} g(J', \rho) = \log Z(G^{(r)}, 0, J^{(r)}_{min}) - (M^{(r)})^\top J^{(r)}_{min}, \,\,\, M^{(r)} \triangleq \frac{\partial}{\partial J^{(r)}_{min}} \log Z(G^{(r)}, 0, J^{(r)}_{min})
\end{equation*}
where $\{ J^{(r)}_{min} \}$ is argmin inside $g(J', \rho)$'s definition and $M^{(r)} = \{ M^{(r)}_e \, | \, e \in E(G^{(r)}) \}$ is a vector of \textit{pairwise marginal expectations}. We reparameterize $\rho(r)$ into $\frac{w(r)}{\sum_{r'} w(r')}$ where $w(r) > 0$.

For $e = \{ v, w \} \in E(G)$ we approximate pairwise marginal probabilities as \citet{wainwright,globerson}
\begin{equation*}
    \mathbb{P}^{alg} (x_v x_w = 1) = \frac12 \cdot \lbrack \sum_r \rho(r) M^{(r)}_e \rbrack + \frac12
\end{equation*}
Let $e_{A}$ be an edge between central vertex $v$ and apex in $G'$. We approximate singleton marginal probability at vertex $v$ as
\begin{equation*}
    \mathbb{P}^{alg} (x_v = 1) = \frac12 \cdot \lbrack \sum_r  \rho(r) M^{(r)}_{e_A} \rbrack + \frac12
\end{equation*}

\vskip 0.2in
\bibliography{sample}

\end{document}